\documentclass[%
 reprint,
nofootinbib,
 amsmath,amssymb,
 aps,
]{revtex4-1}

\usepackage{graphicx}
\usepackage{hyperref}

\usepackage{amsthm}
\usepackage{enumitem}
\usepackage{xcolor} 
\usepackage[framemethod=tikz]{mdframed}

\definecolor{warning_bgcol}{RGB}{252,248,229}
\definecolor{warning_textcol}{RGB}{111,89,54}
\definecolor{warning_linecol}{RGB}{248,235,207}
\definecolor{danger_bgcol}{RGB}{239,223,222}
\definecolor{danger_textcol}{RGB}{128,60,57}
\definecolor{danger_linecol}{RGB}{230,205,209}
\definecolor{success_bgcol}{RGB}{224,237,216}
\definecolor{success_textcol}{RGB}{68,104,60}
\definecolor{success_linecol}{RGB}{218,232,201}
\definecolor{info_bgcol}{RGB}{220,237,246}
\definecolor{info_textcol}{RGB}{58,100,126}
\definecolor{info_linecol}{RGB}{196,231,240}

\mdfdefinestyle{box_style}{
skipabove=\baselineskip,
skipbelow=\baselineskip,
innertopmargin=\baselineskip,
innerbottommargin=\baselineskip,
innerleftmargin=.4\baselineskip,
innerrightmargin=.4\baselineskip,
splittopskip =1.5\baselineskip,
splitbottomskip =\baselineskip,
roundcorner=.4\baselineskip
}
\mdfdefinestyle{warning_style}{
style=box_style,
backgroundcolor=warning_bgcol,
linecolor=warning_linecol,
fontcolor=warning_textcol,
}
\mdfdefinestyle{success_style}{
style=box_style,
backgroundcolor=success_bgcol,
linecolor=success_linecol,
fontcolor=success_textcol,
}
\mdfdefinestyle{danger_style}{
style=box_style,
backgroundcolor=danger_bgcol,
linecolor=danger_linecol,
fontcolor=danger_textcol,
}
\mdfdefinestyle{info_style}{
style=box_style,
backgroundcolor=info_bgcol,
linecolor=info_linecol,
fontcolor=info_textcol,
}

\newmdenv[style=warning_style]{warning}
\newmdenv[style=info_style]{info}
\newmdenv[style=danger_style]{danger}
\newmdenv[style=success_style]{success}

\theoremstyle{remark}
\newtheoremstyle{mythmstyle}%
{0pt}
{0pt}
{}
{}
{\bf}
{.}
{.5em}
{}

\theoremstyle{mythmstyle}
\newtheorem{thm}{Theorem}
\newtheorem{lem}[thm]{Lemma}

\surroundwithmdframed[style=info_style]{thm}
\surroundwithmdframed[style=info_style]{lem}
\surroundwithmdframed[style=info_style]{cor}
\surroundwithmdframed[style=info_style]{prop}
\surroundwithmdframed[style=warning_style]{defn}

\newcommand{\skproof}{Sketch of proof}
\makeatletter
\renewenvironment{proof}[1][\proofname]{\par
\pushQED{\qed}%
\normalfont \topsep6\p@\@plus6\p@\relax
\trivlist
\item\relax
{\bf
#1\@addpunct{.}}\hspace\labelsep\ignorespaces
}{%
\popQED\endtrivlist\@endpefalse
}
\makeatother

\begin{document}

\title{2D quantum computation with 3D topological codes}
\author{H\'ector Bomb\'in}
\affiliation{PsiQuantum, Palo Alto}
 
\begin{abstract}
I present a fault-tolerant quantum computing method for 2D architectures that is particularly appealing for photonic qubits. It relies on a crossover of techniques from topological stabilizer codes and measurement based quantum computation. In particular, it is based on 3D color codes and their transversal operations.
\end{abstract}

\maketitle

\section{Introduction}

It is well appreciated that a programmable quantum computer can only be constructed by using methods from the theory of quantum fault tolerance to deal with the noise that arises at all stages of the computation%
\footnote{
That is, unless we can find components that function reliably enough in the absence of any active error correction~\cite{brown:2016:quantum}.
}
~\cite{lidar:2013:quantum}. 
Topological stabilizer codes are the most promising route to achieving full fault-tolerant operation in the near term~\cite{campbell:2017:roads}. A major drawback of these codes, however, is that in two-dimensional architectures the native fault-tolerant gates are not a universal gate set for quantum computing. There are ways to get around this, the most prominent being magic state distillation (MSD)~\cite{bravyi:2005:universal}. MSD relies on building a much larger quantum computer and then using teleportation of states produced in the extra "magic state factories" to do the gates required for full universality. Despite considerable effort and progress in optimizing MSD techniques, the overheads are still extremely high, and the resources in terms of numbers of qubits and gates used for the MSD actually dominate the resource costs of doing the computation~\cite{campbell:2017:roads}.

Alternatively, it is known that by moving to three-dimensional architectures MSD can be avoided, and a universal gate set can be performed natively~\cite{bombin:2016:dimensional}. However, there are considerable practical problems with building a 3D array of interacting qubits.

The main result of this paper is that a native universal set of gates \emph{is} possible in a scalable 2D physical architecture. The result is built around the unique perspective offered by measurement based quantum computing (MBQC)~\cite{raussendorf:2001:one}, and is both inspired by and well-tailored to photonic quantum computing~\cite{rudolph:2017:optimistic} because its simplest realization makes use of the natural ability to delay photons.

\subsection{A dimensional puzzle}

Topological error correction, originally introduced in the foundational work of Kitaev~\cite{kitaev:1997:quantum}, is likely to play an important role in the ongoing technological race to build a quantum computer. This is particularly true for topological methods lying at the crossover with stabilizer-based approaches, a sweet spot where the versatility and simplicity of stabilizer techniques~\cite{gottesman:1996:stabilizer} joins hands with the locality and scalability of topological error correction~\cite{dennis:2002:tqm}.

The qubits of a topological stabilizer code form a lattice such that (i) the quantum operations required for error correction are highly localized, whereas (ii) logical information takes the form of delocalized degrees of freedom dependent on the overall topology of the lattice~\cite{bombin:2013:topological}.
There exist two intrinsic methods to compute with these topological degrees of freedom. The first one, code deformation, relies on modifying the topology of the system over time~\cite{dennis:2002:tqm, raussendorf:2007:deformation, bombin:2009:deformation, bombin:2010:twist, horsman:2012:surface, landahl:2014:quantum,yoder:2017:surface}. Unfortunately the resulting encoded operations are constrained to the Clifford group, and thus have to be supplemented with costly~\cite{campbell:2017:roads} magic state distillation~\cite{bravyi:2005:universal} to achieve universal computation.

The second option is using transversal gates%
\footnote{
Transversal gates are by no means unique to topological methods. However, in the topological context they naturally generalize to finite depth quantum circuits built out of gates involving a few neighboring qubits each~\cite{bravyi:2013:classification}. It is this generalized perspective that makes transversal gates `natural' for topological codes.
}.
Remarkably, the set of encoded gates achievable with transversal gates becomes less constrained as the number of spatial dimensions of the code grows~\cite{bombin:2013:self, bravyi:2013:classification}. Two-dimensional codes, the most interesting from a practical perspective, are also the most constrained: only Clifford operations are feasible. In three dimensions, by contrast, local operations, supplemented with global classical computation, are enough to achieve universal quantum computation. This is true in particular for color codes, a class of topological stabilizer codes with optimal transversality properties for every spatial dimensionality~\cite{bombin:2006:2dcc, bombin:2007:3dcc, bombin:2013:self, bombin:2015:gauge, bombin:2018:transversal}.

The three-dimensional scenario has further advantages. Not only it is possible to compute fault-tolerantly by purely topological means, but also all elementary encoded operations can be carried out by means of finite depth compositions of geometrically local gates, supplemented with global classical computation. This is made possible by a technique known as single-shot error correction~\cite{bombin:2015:single-shot, campbell:2018:theory}, in particular for a so-called `gauge' variant of three-dimensional color codes~\cite{bombin:2015:gauge, bombin:2016:dimensional}.

The state of affairs just presented is puzzling~\cite{campbell:2017:roads}: for three spatial dimensions a universal set of topological operations can be carried out in constant time,  and yet for two spatial dimensions universality cannot be achieved even if operations extend over time, despite the fact that
\begin{equation}
3+0 = 2+1.
\end{equation}

\subsection{Colorful quantum computation}
 
This paper introduces a purely topological approach to fault-tolerant quantum computation that is based on stabilizer codes, and yet is scalable for just two spatial dimensions. 
Unlike in constructions based on the toric code~\cite{raussendorf:2005:long, raussendorf:2006:ftoneway, raussendorf:2007:deformation, raussendorf:2007:fault}, logical gates
\begin{itemize}
\item
form a universal set, and 
\item
are based on transversal operations, rather than code deformation%
\footnote{As in~\cite{knill:1996:threshold, bombin:2007:3dcc}, the use of transversal measurements sidesteps the fact that transversal gates on a given error-correcting code are never universal~\cite{eastin:2009:restrictions}.
}.
\end{itemize}
In contrast with the non-scalable methods of~\cite{bravyi:2015:doubled, jochym:2015:stacked}, the new approach, colorful quantum computation,  involves no new error correcting codes. Conventional 3D color codes suffice%
\footnote{
Gauge color codes are an original motivation for the result but not part of it.
}.
It is presented below in three different but closely related forms
\begin{itemize}
\item
a 3D MBQC scheme (section~\ref{sec:3D}), 
\item
a hybrid scheme combining 2D computation and information storage over a third dimension (section~\ref{sec:hybrid}),
\item
a scalable 2D scheme reliant on a new form of decoding (sections~\ref{sec:2D} through~\ref{sec:naive}). 
\end{itemize}

The three-dimensional MBQC scheme is obtained by encoding each of the qubits of a regular MBQC scheme in a tetrahedral code, a class of 3D color codes. The construction is enabled by the following ingredients: 
\begin{itemize}
\item
Single-shot initialization in the $X$ basis: this is possible thanks to the fact that 3D color codes, from a condensed matter perspective, are partially \emph{self-correcting}~\cite{bombin:2015:single-shot}.
\item
Transversal Pauli and $X\pm Y$ logical measurements: these are possible thanks to the transversal T gate and the CSS~\cite{calderbank:1996:good,steane:1996:multiple} structure of tetrahedral color codes.
\item
Transversal CP gates that involve only the qubits at the two-dimensional contact region of two tetrahedral code lattices: like the closely related dimensional jumps~\cite{bombin:2016:dimensional}, these are possible thanks to the matryoshka-like nature of tetrahedral codes and their higher dimensional analogues~\cite{bombin:2018:transversal}.
\end{itemize}

A straightforward rotation of the $(3+0)$-dimensional MBQC scheme to make it $(2+1)$-dimensional is not compatible with the causal structure of single-qubit measurements. Two different methods overcome this difficulty:
\begin{itemize}
\item
A non-scalable hybrid scheme where computation is two-dimensional and a third spatial dimension is used for information storage. The purpose of storage is to delay some of the measurements. Since qubits are stored for a fixed amount of time and not accessed in between, a natural implementation are photonic qubits delayed on optical fiber.
\item
A scalable two-dimensional scheme that relies on `just-in-time' (JIT) decoding to satisfy the causal constraints. Unlike the conventional decoding used in the three-dimensional scheme, JIT decoding corrects the outcomes of single-qubit measurements as they become available.
\end{itemize}

\begin{danger}
\center {\bf Notation} is listed in appendix~\ref{sec:notation}.
\end{danger}

\section{Overview}\label{sec:overview}

This section discusses some key points driving the results below.

\subsection{Why tetrahedral codes?}

An operation on a many-body system is \emph{quantum local} when it involves a local quantum circuit (\emph{i.e.} of finite depth and possibly composed of geometrically local gates, as in this context) assisted with non-local classical computation~\cite{bombin:2015:single-shot}. It is a natural many-body analogue of LOCC.

Tetrahedral codes are a class of 3D topological codes that encode a single logical qubit. The qubits of a tetrahedral code form a lattice with the overall topology of a tetrahedron, hence the name. They have an exceptional set of quantum-local gates~\cite{bombin:2018:transversal}.

\begin{success}
The following operations are quantum local in tetrahedral codes:
\begin{itemize}
\item
preparations in the $X$ basis, 
\item
controlled phase (CP) gates, 
\item
Pauli and $X\pm Y$ measurements.
\end{itemize}
Locality is geometric in a 3D setting if CP gates only involve logical qubits residing on adjacent tetrahedra.
\end{success}
\noindent
These are the operations required in MBQC~\cite{raussendorf:2001:one}. This suggests performing MBQC with each of the qubits of the resource state encoded in a tetrahedral code. As long as the MBQC scheme only requires its qubits to be online for a bounded period of time (independent of the size of the computation), fault-tolerance can be achieved because~\cite{bombin:2018:transversal}
\begin{itemize}
\item
preparations give rise to errors with a local distribution of syndromes,
\item
the rest of operations are compatible with such errors, and
\item
measurements have built-in error correction.
\end{itemize}
Notice in particular that there is no need to perform fault-tolerant error correction rounds on the tetrahedral codes. At a deeper level, however, fault-tolerant error correction, intertwined with logical gates, is still happening. Its target are the logical qubits of the correlation space picture~\cite{gross:2007:novel}.

\subsection{Why just-in-time decoding?}

The MBQC schemes of interest here are 
\begin{itemize}
\item
tightly connected to the circuit model, and 
\item
based on a two- or three-dimensional graph state.
\end{itemize}
In particular, one of the dimensions of the lattice corresponds to time in the equivalent circuit model, and it is possible to order operations so that at any given time the qubits that are online form a one- or two-dimensional sublattice.

This sublattice becomes three-dimensional when the qubits of the graph state are encoded in tetrahedral codes. The question then is whether one of these three dimensions can be made time-like to recover a setting with just two spatial dimensions. As it turns out this is not possible under the conventional operation of trahedral codes, because causality is not preserved.

The origin of the problem is that the preparation of logical states is quantum-local, rather than local. It is, however, a local operation if prepared state are correct only up to some Pauli operator, a so called Pauli frame. The obstruction comes from non-Pauli logical measurements: in contrast with logical Pauli measurements, the Pauli frame cannot be processed after measurements happen. In particular, dividing the Pauli frame into an $X$ and a $Z$ component, the causal obstruction can be phrased as follows.
\begin{danger}
\center Logical $X\pm Y$ measurements require the $X$ component of the Pauli frame.
\end{danger}
\noindent
A way out of this obstruction is to compute the Pauli frame piecewise, as information becomes available. This can be done, and in a fault-tolerant manner. The key ingredient, JIT decoding, can be carried out by means of (repurposed) standard decoders. Moreover the discrepancies between the Pauli frames obtained via JIT and conventional decoding can be transformed into (mostly) erasure errors. This is a welcome feature because the JIT-decoded version is unavoidably more noisy, being limited by causality.

\subsection{Why photonic qubits?}

An alternative method to deal with the causal obstruction discussed above is delaying the single qubit measurements that compose each logical $X\pm Y$ measurement. Since the delay time is fixed beforehand for a given code size, photonic qubits on optical fiber are perfectly fit for the task. Moreover, photons move from end to end of the optical fiber as they are stored, preserving the locality of operations in a three-dimensional setting. Notice that the third dimension is only required for optical fiber, whereas the computational part of the scheme, including the measurements of delayed qubits, remain two-dimensional.

There is another way in which photonic qubits are a good fit. Even though MBQC and conventional (circuit based) quantum computation are equivalent, it is apparent from the discussion above that the MBQC picture is more natural for the schemes presented here, a feature that is shared with photonic quantum computation~\cite{rudolph:2017:optimistic}.
 
\section{3D scheme}\label{sec:3D}

This section introduces colorful quantum computation in its simplest incarnation: as a 3D MBQC scheme obtained by encoding the qubits of an ordinary MBQC scheme with tetrahedral codes. In contrast to the 3D scheme of~\cite{raussendorf:2007:deformation}, fault-tolerance is achieved here by purely topological means, thus eliminating the need for magic state distillation~\cite{bravyi:2005:universal}.

Sections~\ref{sec:hybrid} and~\ref{sec:2D} discuss two different methods to make one of the tree spatial dimensions of this scheme time-like.

\subsection{Tetrahedral codes}

A tetrahedral colex~\cite{bombin:2007:3dcc} is a lattice with the overall topology of a tetrahedron. Each of its facets and (3-)cells is labeled with one of four colors (red, green, blue and yellow), in such a way that each vertex belongs to exactly one cell or facet of any given color, see~\cite{bombin:2018:transversal}.

For each tetrahedral colex there is a tetrahedral code~\cite{bombin:2007:3dcc}: a stabilizer code~\cite{gottesman:1996:stabilizer} with
\begin{itemize}
\item
a qubit per vertex of the colex,
\item
a stabilizer generator in $P_X$ per cell (called cell operator),
\item
a stabilizer generator in $P_Z$ per face (called face operator).
\end{itemize}
Typically we are not interested on a single code, but rather on a family with a fixed local lattice structure. The code distance is proportional to the lattice size, and it is desirable for it to increase indefinitely within the family, a constant at a time. An example of such a family is given in~\cite{bombin:2015:gauge}.


\subsection{Encoding}

The aim is to encode the qubits of a MBQC scheme to make it fault-tolerant. We need to fix some terminology.
\begin{warning}
\begin{itemize}
\item
{\bf Logical state:} the resource state of the MBQC scheme.
\item
{\bf Encoded state:} the logical state encoded with tetrahedral codes.
\item
{\bf Resource state:} the precursor of the encoded state.
\end{itemize}
\end{warning}
\noindent
The logical and resource states are conveniently represented as graph states, and thus we talk about logical and resource graphs. Logical qubits are the qubits of the logical state. They should not be confused with the logical qubits of the equivalent circuit model, which do not enter the discussion.

\subsubsection{Resource state}

\begin{figure}
\centering
\includegraphics[width=.9\columnwidth]{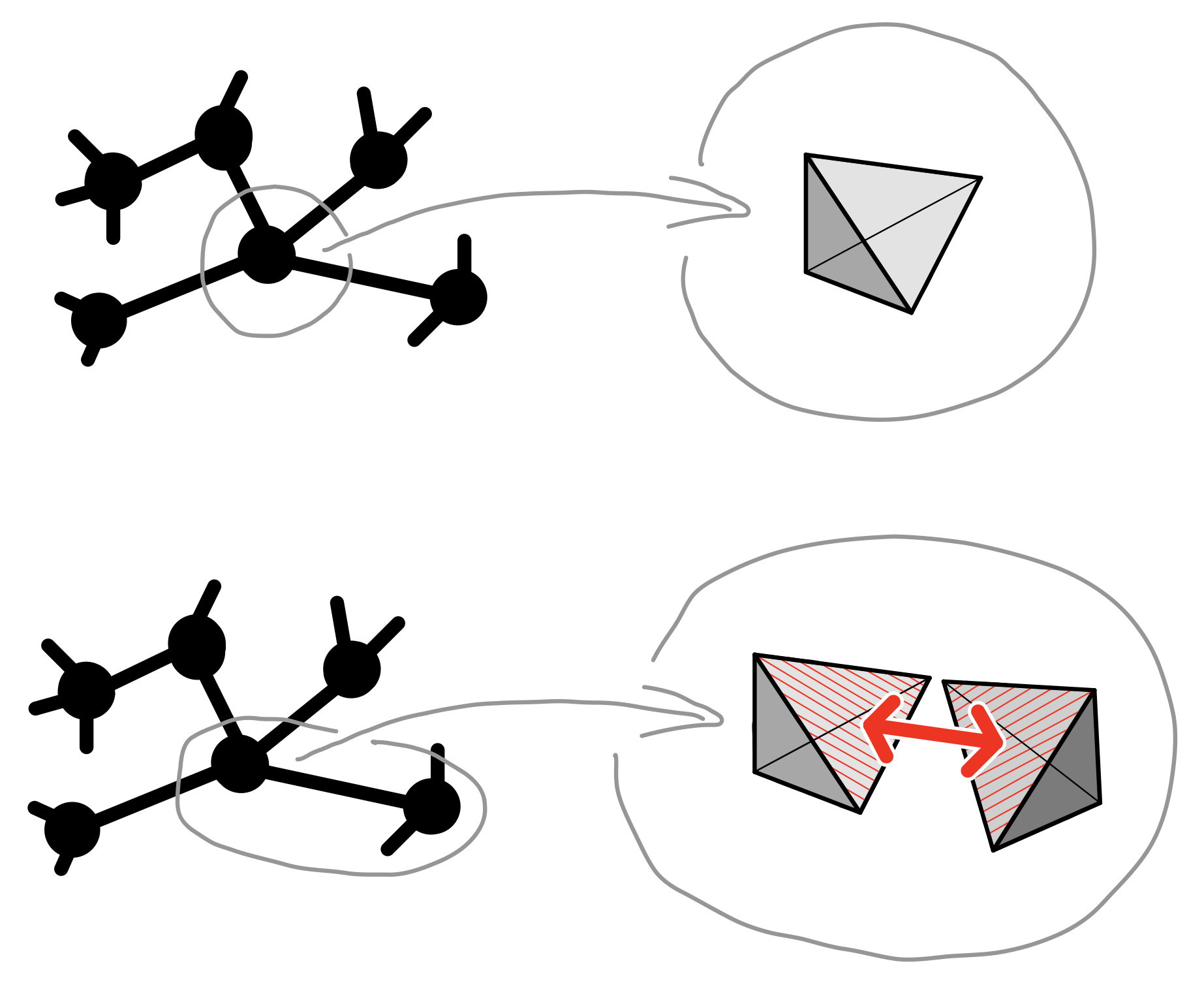}
\caption{(Top) Each vertex of the logical graph contributes a tetrahedral colex. (Bottom) Each edge contributes a pair of matched facets.}
\label{fig:encoding}
\end{figure}

Consider an arbitrary logical graph. The geometry of the resource graph is dictated by a collection of tetrahedral colexes with matched facets, see figure~\ref{fig:encoding}:
\begin{itemize}
\item
There is a tetrahedral colex per vertex of the logical graph.
\item
If two vertices are linked on the logical graph, the corresponding tetrahedral colexes have a \emph{matched} facet pair.
\end{itemize}
The matching establishes a one-to-one relation between the vertices of the two facets that induces a one-to-one relation between edges and faces%
\footnote{In particular, matched facets have the same geometry. A given facet can in principle be matched to an arbitrary number of other facets \emph{but} (i) fault tolerance requires the number to be bounded, and (ii) if tetrahedra do not overlap the number is at most one.}.

\begin{figure}
\includegraphics[width=\columnwidth]{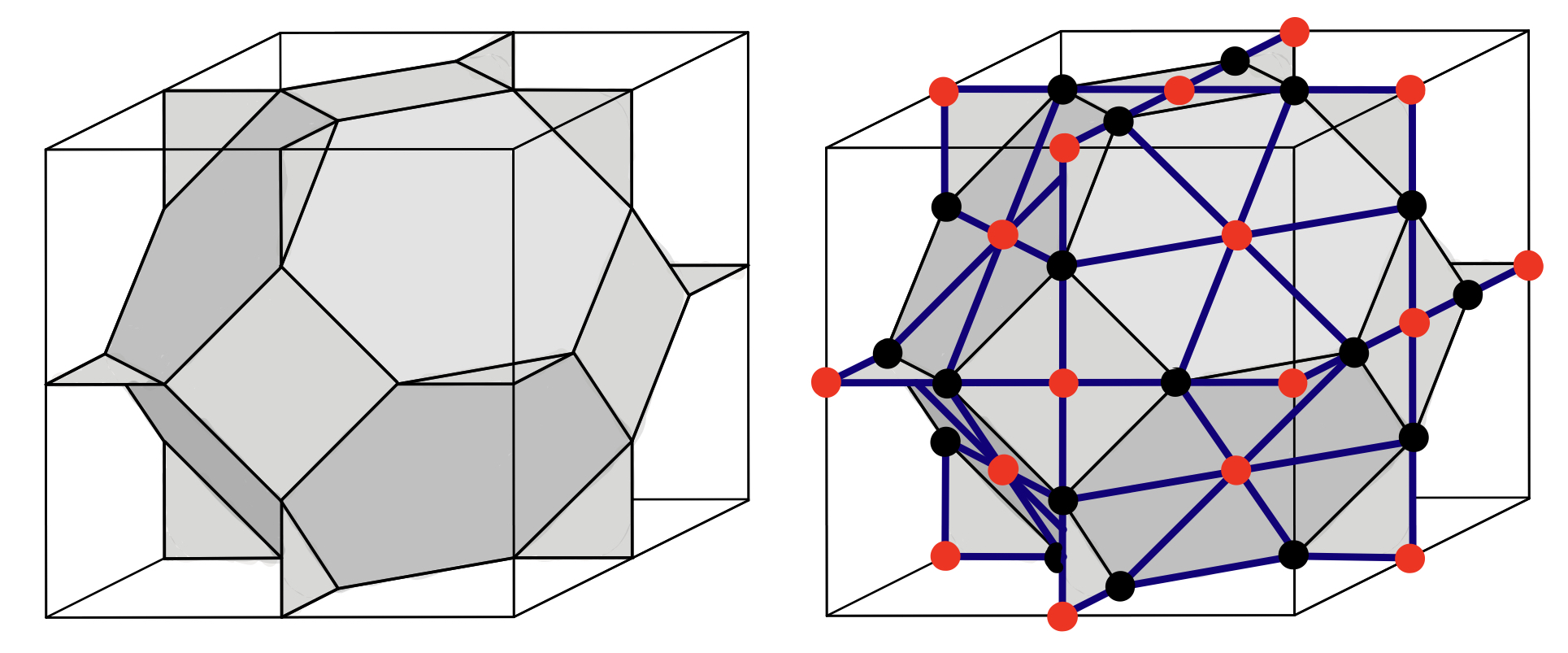}
\caption{The unit cell for the colex family in~\cite{bombin:2015:gauge} (left) and the corresponding resource graph unit cell (right). Code qubits are displayed in black, ancilla qubits in red, and inner edges in purple. Notice that, within the bulk, resource graph vertices are 4- and 6-valent.
}
\label{fig:unit}
\end{figure}

\begin{warning}
The resource state has two kinds of qubits,
\begin{itemize}
\item
{\bf code qubits}, one per vertex of the tetrahedral colexes, and
\item
{\bf ancilla qubits}, one per face of the tetrahedral colexes,
\end{itemize}
and two kind of edges in its graph,
\begin{itemize}
\item
{\bf inner edges}, connecting each ancilla qubit to each of the code qubits at the ancilla's face, and
\item
{\bf outer edges}, connecting matched code qubits.
\end{itemize}
\end{warning}
\noindent
Figure~\ref{fig:unit} illustrates the resource state in the bulk of a colex. 

\subsubsection{Encoded state}\label{sec:encoded}

Single-qubit measurements provide the link between the trivial entanglement%
\footnote{In condensed matter terms.} 
of the resource state and the global entanglement pattern~\cite{chen:2010:local} that chacterizes the encoded state.
\begin{success}
\center Ancilla qubits are always measured in the $X$ basis. 
\end{success}
\noindent
The result of measuring the ancillas is the encoded state \emph{up to a Pauli frame} that depends on the ancilla outcomes and is discussed below%
\footnote{The Pauli frame makes the preparation of the encoded state quantum-local, rather than local, sidestepping the impossibility of connecting distinct topologically ordered phases with local operations~\cite{chen:2010:local}. Notice that the entanglement pattern of the encoded state is trivial with respect to quantum-local operations, which are the condensed matter analogue of LOCC. This triviality under quantum-local operations is a distinguishing feature of abelian topological order that, as illustrated here, has important applications for topological fault tolerance. 
}. 
Reorganizing the operations makes this evident%
\footnote{See~\cite{bombin:2018:transversal} for the necessary background on tetrahedral codes.
}:
\begin{enumerate}
\item
Initialize logical qubits to $|+\rangle$:
\begin{itemize}
\item
Initialize code qubits to $|+\rangle$.
\item
Measure face operators, \emph{i.e.}
\begin{itemize}
\item
initialize ancilla qubits to $|+\rangle$, and
\item
apply a CP gate per inner edge.
\end{itemize}
\end{itemize}
\item
Apply a logical CP gate per logical edge: 
\begin{itemize}
\item
Apply a CP gate per outer edge.
\end{itemize}
\end{enumerate}

\subsubsection{Encoded measurements}

MBQC on the encoded state proceeds normally. Each logical measurement consists of single-qubit measurements on the code qubits followed by postprocessing:
\begin{success}
\begin{itemize}
\item
If the logical measurement is in a Pauli basis (X, Y or Z), all the code qubits of the tetrahedron are measured in that same basis.
\item
If the logical measurement is in either of the $X\pm Y$ bases, each code qubit is measured in either of those bases.
\end{itemize}
\end{success}
\noindent
Specifically, in the second case the basis choice depends on \cite{bombin:2015:gauge}:
\begin{itemize}
\item
the logical basis: $X+Y$ or $X-Y$,
\item
the position of the qubit in the tetrahedron, and
\item
the outcomes of the ancilla qubits of the tetrahedron.
\end{itemize}

\subsection{Pauli frame}

The encoded graph state is subject to a Pauli frame that is dictated by the ancilla measurement outcomes. Let these outcomes be represented by the set $\phi$ of ancilla qubits (or equivalently 3-colex faces) with negative outcomes. We regard $\phi$ as a flux configuration, see~\cite{bombin:2018:transversal}. The flux configuration $\phi$ is random, but not completely: in the absence of errors it is an $X$-error syndrome (of the tetrahedral codes)%
\footnote{A flux configuration is an error syndrome iff it satisfies a Gauss law, see~\cite{bombin:2018:transversal}.}.

Upon completion of step 1 of section~\ref{sec:encoded}, the result is a collection of encoded $|+\rangle$ states up to a Pauli frame $F_X\in P_X$ with syndrome $\phi$. Any such Pauli frame can be used, the choice is immaterial. Upon completion of step 2 the Pauli frame has propagated across the CP gates, so that the final Pauli frame is
\begin{equation}\label{eq:frame}
F = F_XF_Z, \qquad F_Z\in P_Z,
\end{equation}
where $F_Z$ has support on those qubits that are matched to an odd number of qubits where $F_X$ has support.

\subsubsection{Classical information flow}\label{sec:flow}

We denote by
\begin{equation}
f = f_Xf_Z , \qquad f_X\in P_X,\,\, f_Z\in P_Z,
\end{equation}
the Pauli frame on a given tetrahedron, in contrast with the Pauli frame~\eqref{eq:frame} for the whole system. At each tetrahedron, the flow of classical information in connection with the Pauli frame proceeds as follows:

\begin{itemize}
\item
The ancilla outcomes are processed to obtain $f_X$.
\item
$f_Z$ is obtained from the neighboring tetrahedra.
\item
If a logical Pauli measurement is performed, the code qubit measurement outcomes are processed together with $f$ to produce a logical outcome ($X$ measurements require $f_Z$, $Z$ measurements require $f_X$ and $Y$ measurements require both).
\item
If a logical $X\pm Y$ measurement is performed, $f_X$ is used to choose the measurement basis for each code qubit, and the outcomes are processed together with $f_Z$ to produce a logical outcome.
\end{itemize}
The following causal constraints emanate from this information flow:
\begin{danger}
\begin{itemize}
\item
Within a tetrahedron measured in the $X\pm Y$ basis, all the ancilla qubits have to be measured before the code qubits are measured. 
\item
The following qubits have to be measured before the logical outcome at a given tetrahedron can be computed:
	\begin{itemize}
	\item
	all its code qubits,
	\item
	all its ancilla qubits (except for logical X measurements),
	\item
	all the ancilla qubits of its (matched) neighbors (except for logical $Z$ measurements).
	\end{itemize}
\end{itemize}
\end{danger}

\subsection{Fault tolerance}\label{sec:ft}

For simplicity, we model noise in the system as follows.
\begin{itemize}
\item
The (otherwise ideal) resource state is subject to a local distribution of Pauli errors, see~\eqref{eq:local} below, 
\item
measurements are ideal, but the above errors might depend on the measurement basis, and
\item
classical computation is flawless.
\end{itemize}
Such an error model is meaningful only if 
\begin{itemize} 
\item
in the logical MBQC scheme, qubits are online for a bounded period of time, and
\item
the valence of the vertices of the resource graph is bounded.
\end{itemize}
\begin{warning}
A distribution of error operators, each with support on a set of qubits $W$, is {\bf local} with error rate $0\leq p\leq 1$ if for any set of qubits $Q$
\begin{equation}\label{eq:local}
\text{prob}(Q\subseteq W)\leq p^{|Q|}.
\end{equation}
\end{warning}

\noindent
The preparation of the encoded graph state 
is the crux of the scheme's fault-tolerance. This is due to its quantum-local nature: it only involves local quantum gates, but the Pauli frame is obtained via global classical computation. There is no reason for the resulting residual noise to follow a local distribution: It is necessary to check that it takes a form that can be handled downstream, when code qubit outcomes are processed.

Stochastic Pauli errors amount to stochastic bit-flip errors on the classical ancilla outcomes. Instead of the correct syndrome $\phi$ the noisy outcome is 
\begin{equation}
\tilde\phi = \phi+\omega,
\end{equation}
where $\omega$ is the set of bit-flip locations, which follow a local distribution. Since $\tilde\phi$ is not, in general, a syndrome, it has to be corrected to produce a syndrome $\bar \phi$. The syndrome $\bar\phi$ determines the Pauli frame $F_{\bar \phi}$. The effective error in the Pauli frame is
\begin{equation}\label{eq:Fw}
F_{\bar\omega} = F_\phi F_{\bar \phi},\qquad \bar\omega := \phi+\bar\phi.
\end{equation}
where $F_{\bar\omega}$ is \emph{any} Pauli frame compatible with the syndrome $\bar\omega$, and $F_\phi$ is the \emph{only} Pauli frame compatible with the syndrome $\phi$ and the above equation. Recall that a Pauli frame $F$ is defined from its $X$ component $F_X$, which is in turn only fixed up to an arbitrary logical operator. This freedom allows adjusting $F_\phi$ to the most desirable value of $F_{\bar\omega}$. 

To recap, $\omega$ follows a local distribution with error rate $p$ and the $X$ component of $F_{\bar\omega}$ can be chosen to be correctable, for any given set of correctable errors. Then if 
the error rate $p$ is below a threshold $p_0$, and 
the corrected syndrome $\bar\phi$ is (efficiently) computed as in~\cite{bombin:2015:single-shot},  
then the syndrome of $F_{\bar\omega}$ is \emph{confined}, and a fault-tolerant regime exists, see~\cite{bombin:2015:single-shot, bombin:2018:transversal}.

\section{Hybrid scheme}\label{sec:hybrid}

The second incarnation of colorful quantum computation is a hybrid scheme: computations are carried out in a two-dimensional setting and quantum information storage requires a third dimension. Optical fiber is particularly well suited to attain this extra dimension. Since stored information is not actively corrected, the fault tolerance of this scheme is not scalable%
\footnote{
This is not to say that active error correction is not possible. It just spoils the simplicity of the setting.
}.

A scalable and storage-free two-dimensional approach is discussed in section~\ref{sec:2D}. Very loosely speaking, there is a `continuum' of schemes between these two, using increasingly longer storage times.

\subsection{Delayed measurements}

\begin{figure}
\centering
\includegraphics[width=.8\columnwidth]{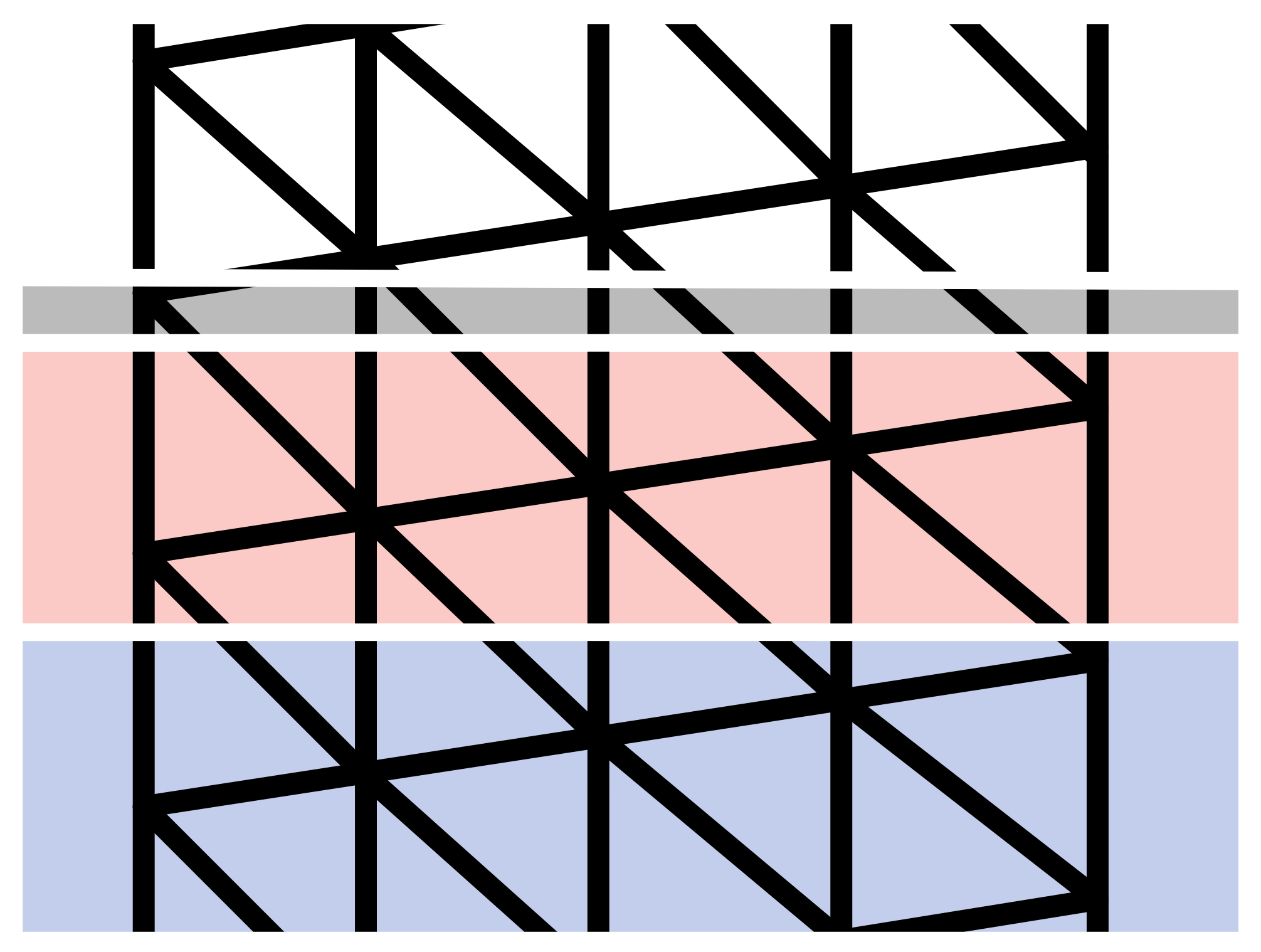}
\caption{
A 2D cut of the 3D MBQC scheme with time as a dimension. Black triangles represent tetrahedra. The colored areas represent regions of the resource state at a point in time. The white one is yet to be created. The grey one is being created.
The red one has been created and its ancillas measured.
The blue one has been created and all its qubits measured.
}
\label{fig:spacetime_delay}
\end{figure}

As discussed in section~\ref{sec:overview}, time cannot be naively exchanged with one of the dimensions of the 3D MBQC scheme of section~\ref{sec:3D}. This is due to causal obstructions. The obstacle lies in the flow of classical information, in particular within logical $X\pm Y$ measurements. The Pauli frame $f_X$ on the tetrahedron is computed using all the ancilla outcomes, and thus, as noted in section~\ref{sec:flow}, code qubits can be measured only after all the ancilla qubits have been measured.

\begin{figure}
\centering
\includegraphics[width=.7\columnwidth]{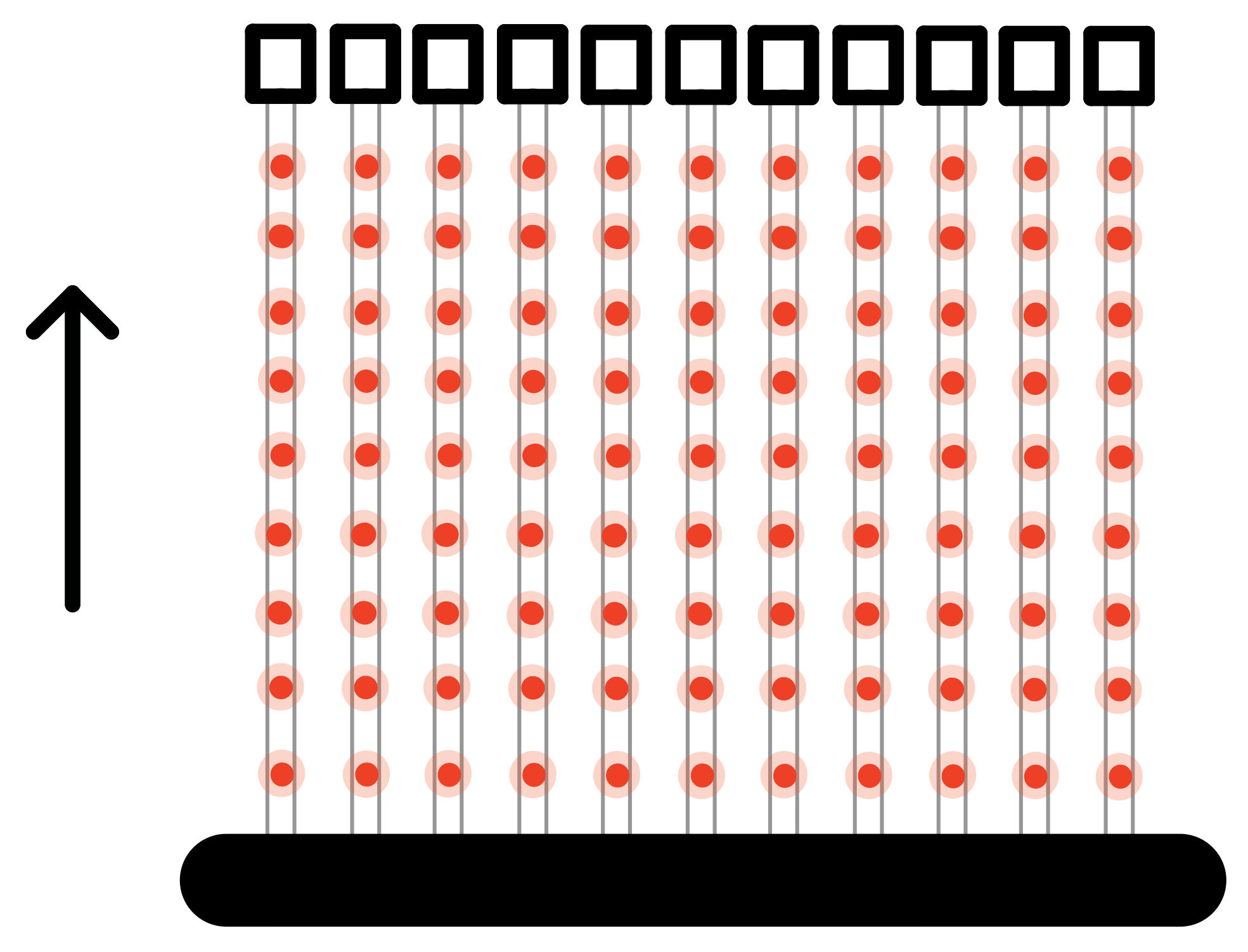}
\caption{The hybrid scheme requires a 2D computational substrate (lower black region) suplemented with information storage to delay some of the single-qubit measurements (boxes on the top). A natural implementation of the delay/storage is optical fiber (vertical lines). The red dots represent photonic qubits at some instant in time.}
\label{fig:delay}
\end{figure}

The spacetime geometry of this constraint is depicted in figure~\ref{fig:delay}. If the code has distance $L$, the information required to choose the measurement basis for a given code qubit is in general available only after a time proportional to $L$. Code qubits have to `sit arround' for that time before being measured. If the number of qubits required for the 2D computation is $N_2$, then the number of online qubits at any given times is
\begin{equation}
N_3\propto LN_2.
\end{equation}
The number of memory qubits required is thus much larger than the number of computational qubits, making the scheme's appeal strongly implementation dependent. In this regard, an all-important aspect is that the quantum memory used in the scheme is of a rather specific kind.

\begin{success}
\center Qubits are stored for a fixed period of time.
\end{success}

\noindent There exists one kind of quantum memory that perfectly matches this scenario: photonic qubits on optical fiber. Figure~\ref{fig:delay} schematically depicts the resulting architecture.
A key feature of optical fiber is that qubits are on the move as they are being stored, which allows to preserve locality in 3D. If the qubits used for storage had a fixed position in time, locality could not be preserved (unless information can be shifted from qubit to qubit locally).

This is not a scalable approach because larger codes requires larger delays. \emph{E.g.} in the case of optical fiber the maximum storage time is likely dictated by optical loss. Once this maximum is reached, the quality of the delay method has to be improved for larger codes to be useful.

\subsection{Concatenation}

In some scenarios the scheme presented here will be concatenated with another scheme capable of \emph{e.g.} dealing with higher levels of noise. In that case the qubits in storage would likely be encoded. A point to note is that although these qubits are to be measured in either of the $X\pm Y$ bases, any other unitarily equivalent basis pair will do, such as X and Z, which are transversal for any CSS code. \emph{E.g.} an (encoded) code qubit could first undergo the sequence of gates  $P\circ H\circ T$, then enter the delay lines, and finally be measured in either the $X$ or $Z$ basis.

\begin{figure}
\centering
\includegraphics[width=.82\columnwidth]{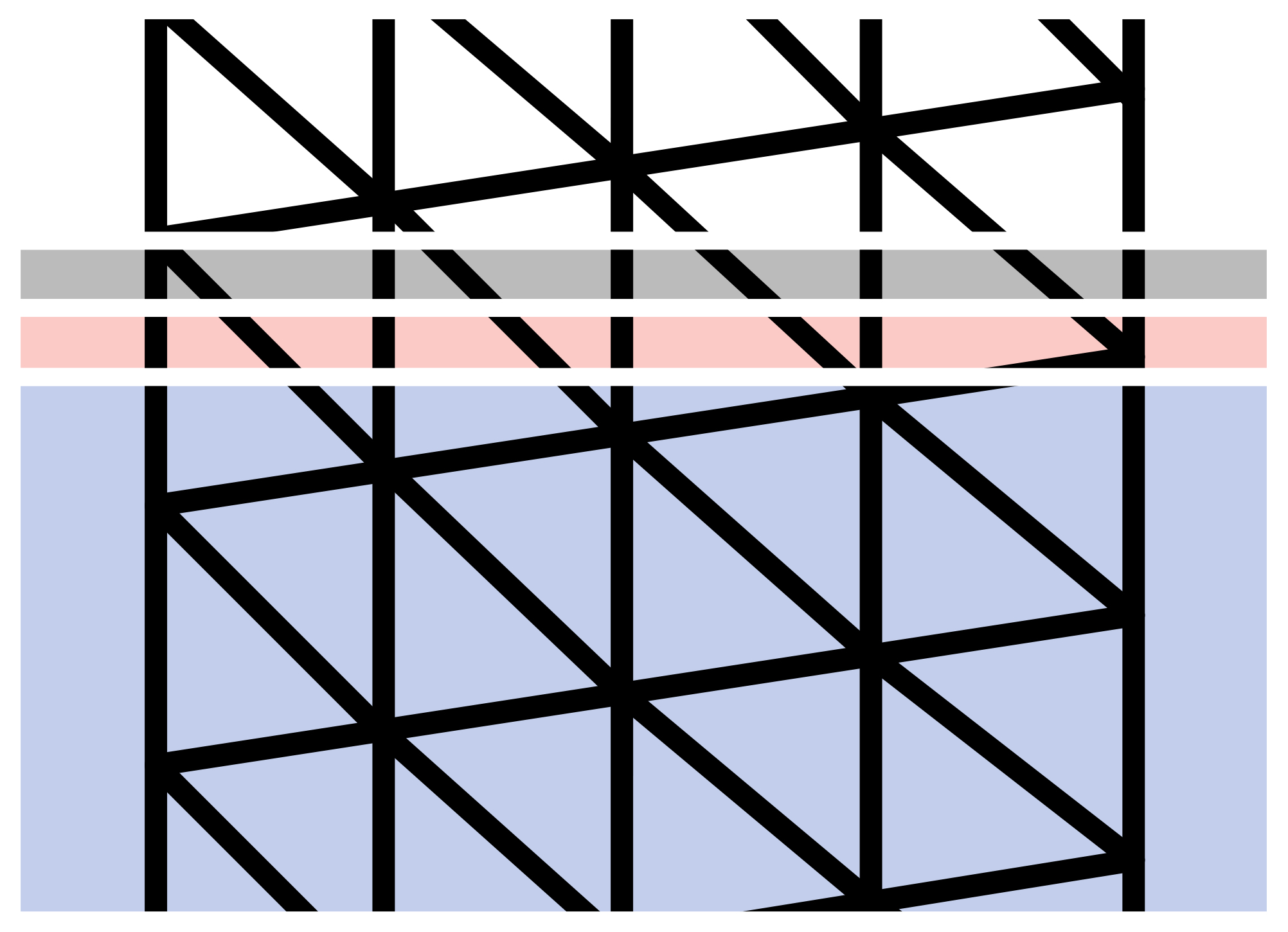}
\caption{This is a variation of figure~\ref{fig:spacetime_delay} in which, at any given time, online qubits are limited to some two-dimensional sublattice of the resource graph.}
\label{fig:spacetime_2D}
\end{figure}

\section{2D scheme}\label{sec:2D}

The third and last incarnation of colorful quantum computation is a purely two-dimensional and scalable scheme. We proceed in two steps, first considering the noiseless case and then adding fault tolerance into the mixture, with just-in-time decoding as a key ingredient. 

Throughout this section two conditions, involving $X$-error syndromes and time, emerge as necessary for the scheme to function. These conditions are the subject of section~\ref{sec:spacetime}, which discusses how they constrain the spacetime geometry of tetrahedra. Section~\ref{sec:twister} describes a compatible architecture.

Just-in-time decoders are only formulated here. Section~\ref{sec:naive} introduces a method to repurpose a conventional decoder for just-in-time decoding.

\subsection{Setting}\label{sec:2D_setting}

In the 3D MBQC scheme, for each tetrahedral code first all the ancilla measurement outcomes $\phi$ are obtained, and subsequently $f_X$, the $X$ component of the Pauli frame, is computed for the whole tetrahedron. We would like to exchange one of the spatial dimensions of the MBQC scheme with time, so that at any given time the qubits that are online form a two-dimensional sublattice, see figure~\ref{fig:spacetime_2D}. To this end the Pauli frame has to be computed as information becomes available.  In particular, the following setting yields a two-dimensional regime%
\footnote{This construction should not be taken too literally. In some systems the circuit model is better suited. The graph state picture is just a convenient theoretical device.
}.
\begin{success}
\begin{itemize}
\item
Each tetrahedron is divided in layers with bounded width in the time direction. The bound is independent of the size of the computation.
\item
The Pauli frame $f_X$ on the code qubits of a given layer is chosen using the ancilla outcomes of that layer and all the preceding ones.
\end{itemize}
\end{success}
\noindent
The following notation is all for a single tetrahedron. On the $i$-th layer, 
\begin{itemize}
\item
$\Lambda_i$ is the set of ancillas/faces, and 
\item
$M_i$ is the set of code qubits.
\end{itemize}
Hence at the $i$-th step the accessible ancilla outcomes correspond to the subset of faces
\begin{equation}
\Phi_i := \bigsqcup_{j=1}^i \Lambda_j,
\end{equation}
and after the $i$-th step the Pauli frame $f_X$ is fixed for the subset of code qubits
\begin{equation}
R_i := \bigsqcup_{j=1}^i M_i.
\end{equation}
The set of all faces is denoted $\Phi$ and the complementary set to $\Phi_i$ is
\begin{equation}
\overline {\Phi_i} := \Phi-\Phi_i.
\end{equation}

\subsection{Causality}\label{sec:causality}

Even if fault tolerance aspects and the time required for classical computation are ignored, a fundamental issue is whether the outlined scenario is compatible with causal constraints at all. Appendix \ref{sec:piecewise} provides the following sufficient condition for causality not to be an obstacle. Let $S_Z$ denote the subgroup of the tetrahedral code stabilizer generated by face operators.
\begin{warning}
{\bf Causality condition.} The eigenvalues of $S_Z\|_{R_i}$ are known at the $i$-th step.
\end{warning} 
As long as the condition holds for every step, it is possible to choose the Pauli frame $f_X$ piecewise. The (rather mild) implications of the causality condition for the spacetime geometry are the subject of section~\ref{sec:spacetime_causality}.

\subsubsection{Logical causality}

Even in the absence of fault tolerance, causality plays a key role in MBQC. In particular, it is not possible to translate directly a given (non-trivial) quantum circuit to a fixed measurement pattern on a graph state. In the original formulation of MBQC it was emphasized that (i) it is possible to fix the measurement basis for all the Pauli measurements, and (ii) on top of those Pauli measurements it is enough to perform $X\pm Y$ measurements. The locations of the $X\pm Y$ measurements is fixed, and only the sign depends on other measurement outcomes, giving rise to causal constraints \cite{raussendorf:2001:one}.

The separation between Pauli and non-Pauli measurements is, however, artificial. It is also possible to fix the measurement basis for all non-Pauli measurements and to adjust instead the basis for some Pauli measurements%
\footnote{From the perspective of the quantum circuit, Pauli and non-Pauli measurements correspond respectively to Clifford and non-Clifford operations. In particular the non-Clifford operations are on the Clifford hierarchy and the techniques of gate teleportation apply~\cite{gottesman:1999:demonstrating}.
}.
This can actually be advantageous in fault-tolerant scenarios, because the position of the variable Pauli measurements can be adjusted (delayed) to allow for decoders to process the non-Pauli measurement. This is not possible in the standard formulation of MBQC because the sign of a $X\pm Y$ measurement on a given qubit depends on the outcomes of its neighbors%
\footnote{
Similar considerations apply in other fault-tolerant scenarios with Pauli frames involved.
}.

Here we adopt the second approach, with fixed non-Pauli measurements, to focus on the problems created by fault-tolerance, rather than logical causality.

\subsection{Just-in-time decoding}\label{sec:JIT}

We are ready to add noise to the picture. In contrast to section~\ref{sec:ft}, here the corrected syndrome $\hat\phi$ has to be generated on the fly, with just a partial knowledge of the noisy ancilla outcomes $\tilde \phi$ and with a limited time to perform the computation. This is just-in-time decoding.

\subsubsection{The challenge}

The goal of JIT decoding is to map, in real time and as information becomes available, the noisy ancilla outcomes $\tilde \phi$ to a corrected syndrome $\hat \phi$. Each step fixes a partial value for the syndrome $\hat \phi$. The input available for this task at the $i$-th step is
\begin{equation}
\tilde \phi\cap \Phi_i.
\end{equation}
For simplicity we assume that this information is used to fix%
\footnote{
This is by no means the only possibility. The set of faces for which $\hat\phi$ is fixed at the $i$-th step could be any subset of $\Phi_i$ compatible with the causality condition, in the sense of~\eqref{eq:causality}.
}
\begin{equation}
\hat \phi\cap \Lambda_i.
\end{equation}
The partial knowledge of the syndrome $\hat\phi$ is then used to choose the Pauli frame $f_X$ in the subset of qubits 
\begin{equation}
R_i-R_{i-1},
\end{equation}
following the prescription of appendix~\ref{sec:piecewise}. The consistency of the procedure is guaranteed if the causality condition is satisfied, \emph{i.e.} 
\begin{equation}\label{eq:causality}
S_Z\|_{R_i}\subseteq S_{\Phi_i},
\end{equation}
where $S_{\phi_i}$ denotes the subgroup of $S_Z$ generated by the faces in $\Phi_i$.

In the conventional decoding of $\tilde \phi$ the relevant classical error correcting code is the set of $X$-error syndromes $\mathcal C$. By contrast, when the known ancilla outcomes are those in the subset $\Phi_i$ the natural code to consider is
\begin{equation}\label{eq:Ci}
\mathcal C_i:=\{
\phi\subseteq \Phi_i
\,|\,
\exists \phi'\subseteq\overline{\Phi_i}\,:\,\phi+\phi'\in\mathcal C
\}.
\end{equation}
The $i$-th step of JIT decoding is not only limited by the lack of knowledge of the future ancilla outcomes, but also by the previous commitment to a final value of $\hat\phi$ in the subset $\Phi_{i-1}$. A straightforward method to overcome this difficulty is the subject of section~\ref{sec:naive}, but it is likely that better ones exist.

\subsubsection{The penalty}\label{sec:delta}

Dealing with the noisy flux configuration $\tilde \phi$ on a layer by layer basis makes correcting errors harder. With full access to $\tilde \phi$ a conventional decoder produces a better estimate $\bar \phi$%
\footnote{
Both $\hat \phi$ and $\bar \phi$ are, of course, dependent on whatever algorithms are used.
}.
Fortunately it is possible to attenuate the damaged caused by using $\hat \phi$, rather than $\bar \phi$, to choose the Pauli frame $f_X$.

\begin{success}
\center Conventional decoding enhances the logical outcome processing.
\end{success}

\noindent In particular, since $\bar \phi$ is the best estimate for the syndrome, it is natural to consider the \emph{differential syndrome}
\begin{equation}
\delta:= \hat \phi + \bar \phi.
\end{equation}
Were $\bar \phi$ correct, using the wrong syndrome $\hat \phi$ would introduce some effective noise. The code qubit measurements can be regarded as a transversal unitary followed by single qubit $X$ mesurements. In this picture, the effective noise amounts to the following errors right before the $X$ measurements%
\footnote{
For further details, including an analysis of the overall effective noise for noisy $\bar\phi$, see~\cite{bombin:2018:transversal}.
}:
\begin{itemize}
\item
a known $Z$ Pauli error that can trivially be compensated for, and 
\item
a random element of a certain group $H(\delta)\subseteq P_Z$ with generators localized in the vicinity of the faces in $\delta$. Such random noise is very similar to qubit erasure, and can be handled analogously.
\end{itemize}
Notice that in the adjacent tetrahedra the Pauli frame should be dictated by $\bar\phi$, rather than $\hat\phi$%
\footnote{
The $X$ component of the Pauli frame for $\bar\phi$ should be chosen so that the $X$ component of the effective noise contributed by $\delta$ is a correctable error, in the sense of~\cite{bombin:2018:transversal}. This is analogous to the discussion in section~\ref{sec:ft}.
},%
to avoid propagating JIT decoding errors between tetrahedra.

\subsection{Closure}

This section discusses a geometrical counterpart to the causality condition, which is of a topological nature. The new condition emerges naturally upon examination of the JIT decoder presented in section~\ref{sec:naive}, where it is required to achieve fault tolerance. However, it is presented here since, unless it is satisfied, it seems unlikely that a JIT decoder can be successful.

Given a familty of encodings for the fault-tolerant scheme under discussion, the following condition needs to be satisfied for some $k_\text{close}$, every tetrahedral colex in the family, and each layer $i$.
\begin{warning}
{\bf Closure condition.} For every $\phi\in\mathcal C_i$ there exists $\phi'\subseteq\overline{\Phi_i}$ such that
\begin{equation}
|\phi'|\leq k_{close} |\phi|, \qquad \phi+\phi'\in \mathcal C.
\end{equation}
\end{warning}
\noindent
Intuitively, the closure condition states that committing to some value of $\hat\phi\cap\phi_i$ should not have disproportionate side effects. If not satisfied, small errors in early stages could get amplified and spoil fault-tolerance.
Section~\ref{sec:spacetime_closure} explores the implications of the closure condition for the spacetime geometry.

\section{Spacetime geometry}\label{sec:spacetime}

This section discusses how to ensure that the causality and closure conditions encountered in section~\ref{sec:2D} are satisfied%
\footnote{
These conditions do not apply to tetrahedra for which the logical measurement is Pauli.
}.
The closure condition imposes significant restrictions to the spacetime geometry of tetrahedral colexes. Section~\ref{sec:twister} introduces a symmetrical architecture in which all tetrahedra satisfy them.

\subsection{Setting}

There are many ways in which the layers of section~\ref{sec:2D_setting} could be defined. For simplicity, we define layers in terms of sets of (3-)cells%
\footnote{This is likely an overkill. A circuit model version should strive for thinner layers.
}.
Namely, for each $i=1,\dots, n$ there is a set of cells $C_i$ with
\begin{equation}
C_1\subset C_2\subset \dots \subset C_n,
\end{equation}
where $C_n$ is the set of all cells. The choice of $C_i$ dictates the content of each layer, and thus the causal structure. In the notation of section~\ref{sec:2D_setting}:
\begin{warning}
\begin{itemize}
\item
The code qubits in $R_i$ are those at a vertex of any cell in $C_i$.
\item
The ancilla qubits in  $\Phi_i$ are those at a face of any cell in $C_i$.
\end{itemize}
\end{warning}

\subsection{Causality condition}\label{sec:spacetime_causality}

\begin{figure}
\centering
\includegraphics[width=.9\columnwidth]{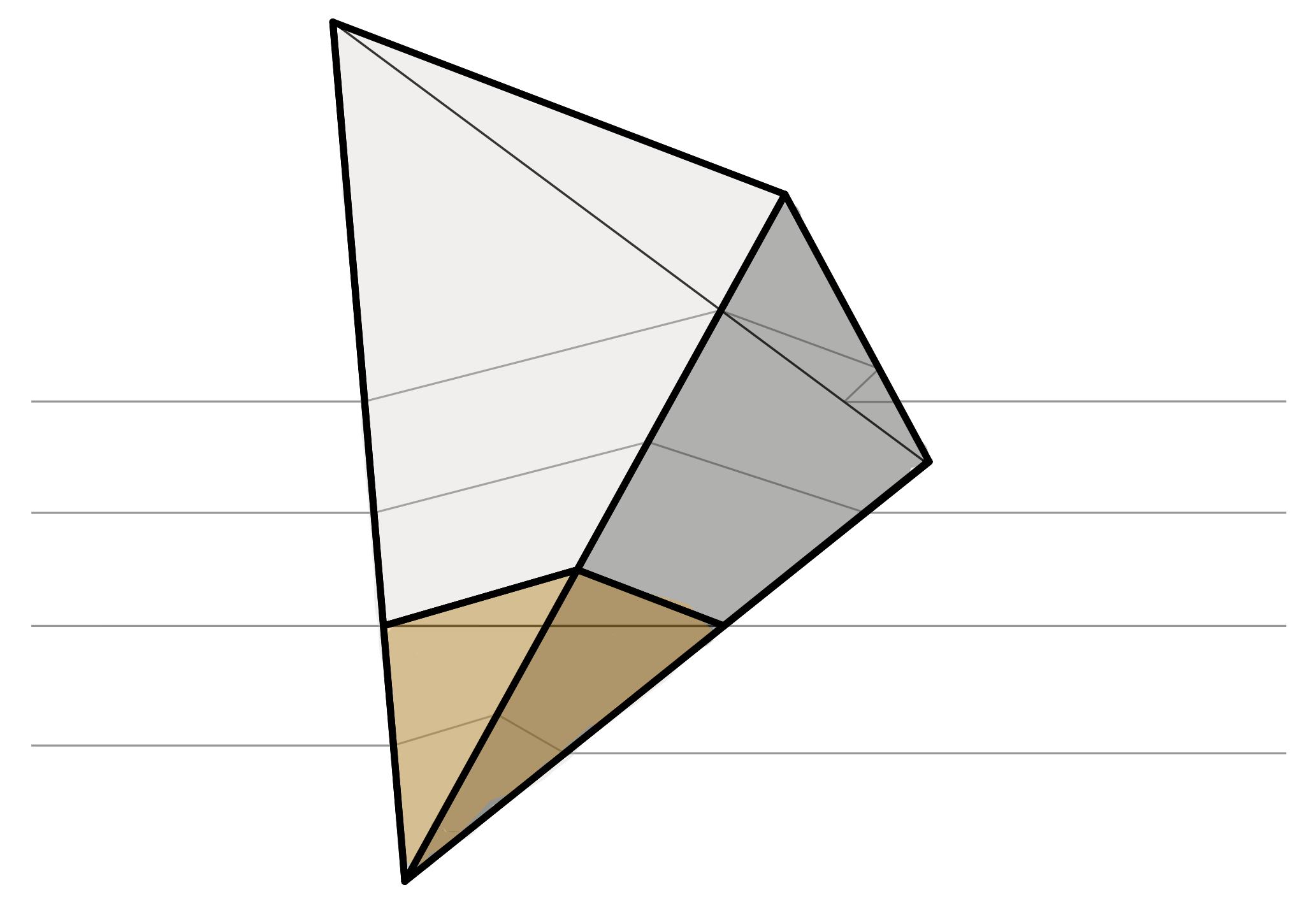}
\caption{Time runs upwards. Horizontal lines are a guide to the eye. The highlighted region of the tetrahedron corresponds to the set of cells $C_i$. The topological constraints imposed by the causality condition are easily satisfied irrespective of the orientation of the tetrahedron.}
\label{fig:spacetime_causality}
\end{figure}

The proof of the following result is in appendix~\ref{sec:ball}.
\begin{lem}\label{lem:causality}
The causality condition is satisfied if 
\begin{itemize}
\item
the submanifold formed by the cells in $C_i$ is a topological ball, and 
\item
its intersection with each cell outside $C_i$, and with each facet of the tetrahedral colex, is empty or a topological disc.
\end{itemize}
\end{lem}
\noindent
The causality condition holds for reasonable geometries. \emph{E.g.} if the facets of the tetrahedra are aproximately flat, see figure~\ref{fig:spacetime_causality}.

\subsection{Closure condition}\label{sec:spacetime_closure}

This condition poses more of a challenge: some care is needed in choosing the spacetime geometry of each tetrahedron.

\subsubsection{Dual graph distances}

\begin{figure}[t]
\centering
\includegraphics[width=.8\columnwidth]{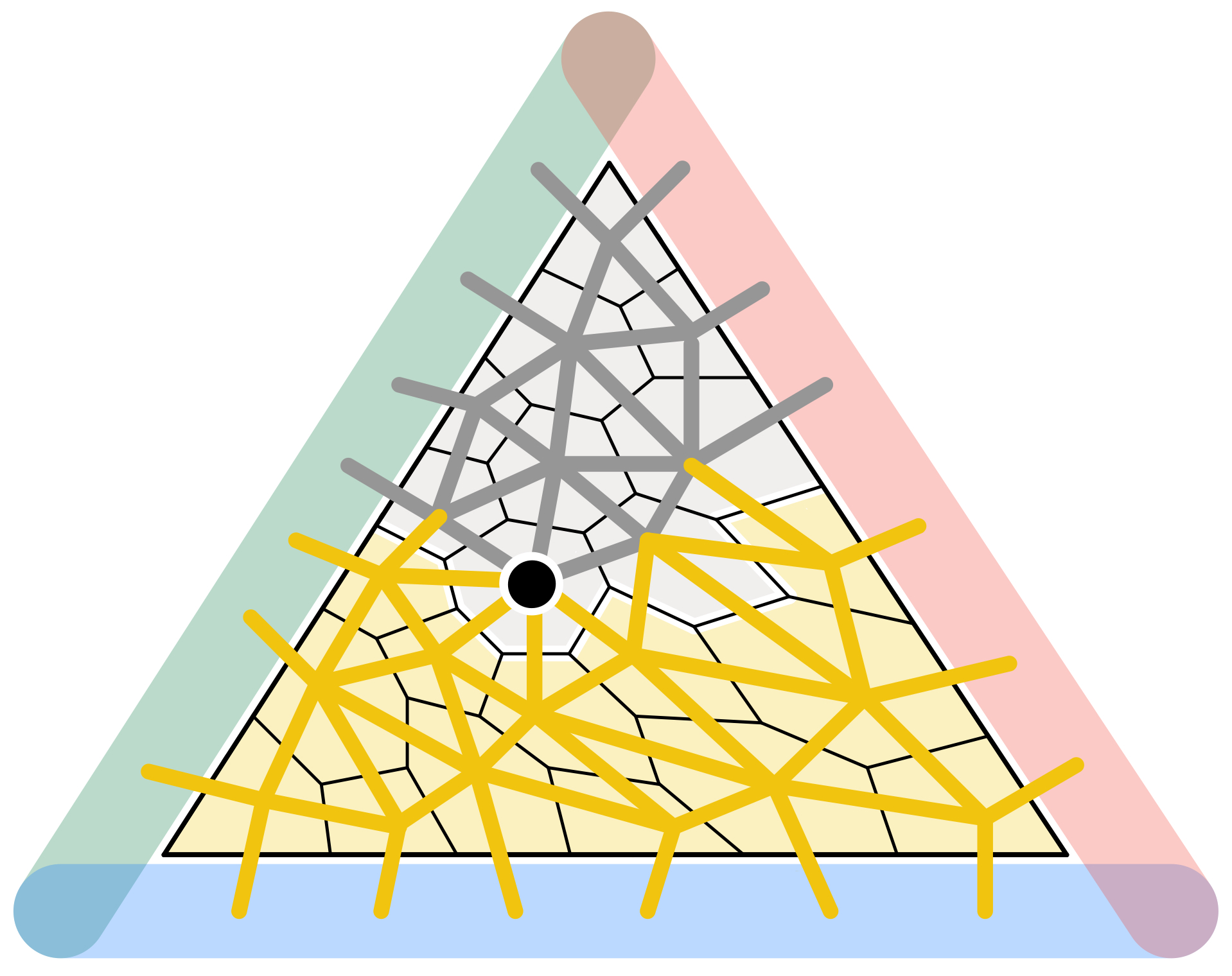}
\caption{
A 2D analogue of the dual graph $\Gamma$. The orange faces of the 2-colex play the role of the set of cells $C_i$ in the 3-colex, and the triangle edges the role of tetrahedron facets. The dual graph $\Gamma_i$ appears in orange, and $\overline{\Gamma_i}$ in grey.  The highlighted vertex $v$ has $d_\text{blue}(v)=3$, $\overline{d_\text{blue}}(v)=\infty$ and $\overline d(v)=5$.
}
\label{fig:gamma}
\end{figure}

The \emph{dual graph $\Gamma$} describes error syndromes in tetrahedral codes~\cite{bombin:2018:transversal}: 
its vertices are dual to the cells of the 3-colex, and its edges are dual to the faces of the 3-colex and represent flux elements. 
Edges that are dual to a face belonging to a facet of the colex have a missing vertex, and we consider them to be \emph{connected to the facet}. 

Let $\Gamma_i$ (respectively $\overline{\Gamma_i}$) denote the subraph of $\Gamma$ with edges dual to the faces in $\Phi_i$ (respectively $\overline{\Phi_i}$). 
For each $i$, we introduce the following distance functions for vertices in $\Gamma_i$%
\footnote{
If no path connects two vertices their distance is infinite.
},
see figure~\ref{fig:gamma}.
\begin{warning}
Given any vertices $v, v'$ of $\Gamma_i$ and a color $\kappa$,
\begin{itemize}
\item
$d(v,v')$ is the distance between $v$ and $v'$,
\item
$d_\kappa(v)$ is the shortest distance from $v$ to any $\kappa$-colored facet, and
\item
$d(v)$ is the sum of the two smallest values of $d_\kappa(v)$ among the 4 values for varying color $\kappa$.
\end{itemize}
Barred versions $\overline d(v,v')$, $\overline d_\kappa(v)$ and $\overline d(v)$ are defined analogously for $\overline{\Gamma_i}$. 
\end{warning}

\subsubsection{Forbidden geometries}

\begin{figure}[t]
\centering
\includegraphics[width=\columnwidth]{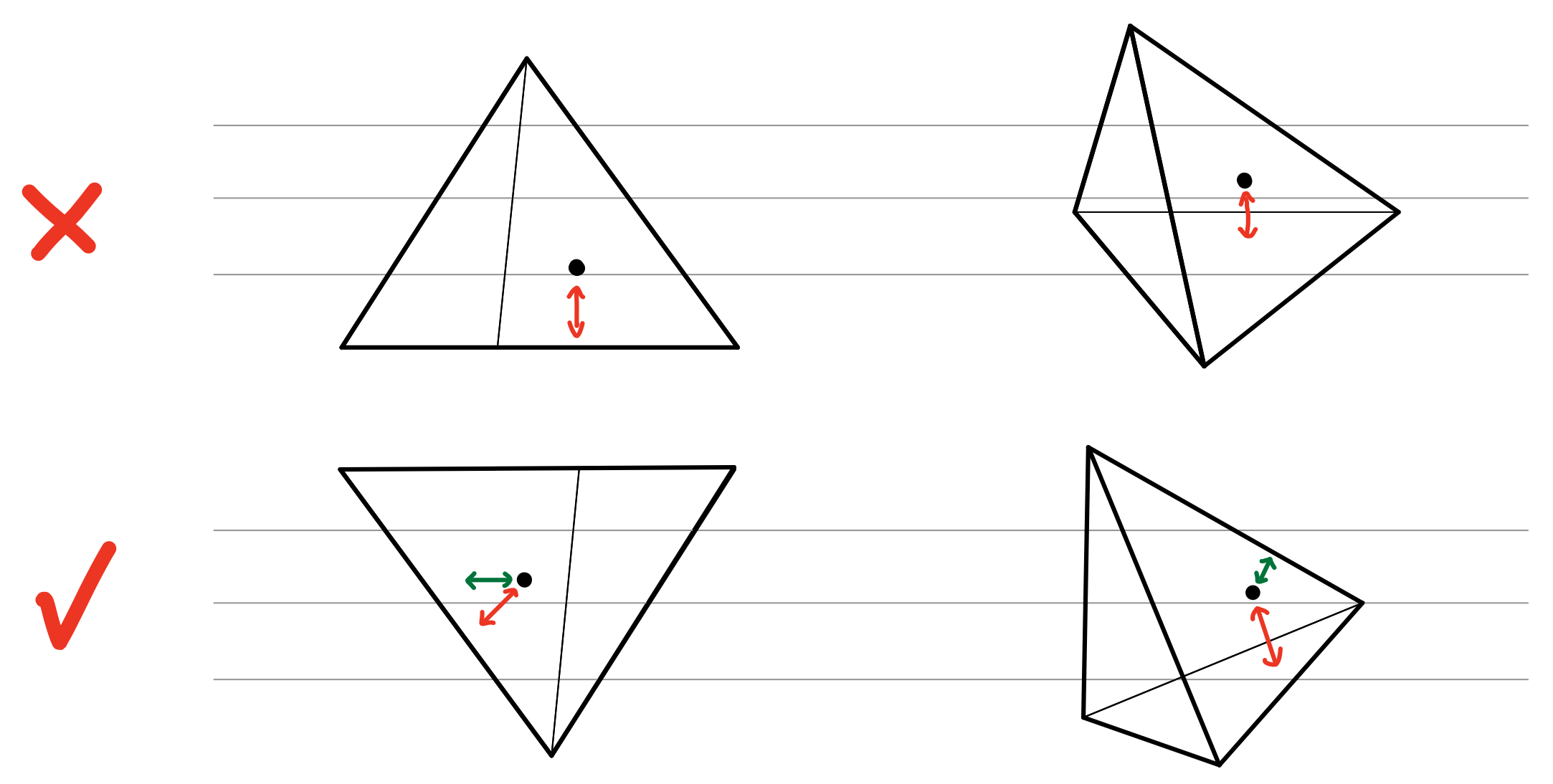}
\caption{
Examples of spacetime geometries and the obstructions created by the closure condition.
Time runs upwards. Horizontal lines are a reminder of the layered structure.
\emph{(Top)} In both examples there are vertices $v$ such that for some $\kappa$ the distance $d_\kappa(v)$ (the red arrow) is very small while $\overline{d_\kappa}(v)$ and $\overline d(v)$ are proportional to the size of the tetrahedron. 
The existence of such vertices is due to the fact that at some point in time a whole facet, or a last remaining edge of it, disappears suddenly. 
\emph{(Bottom)} These geometries are compatible with the closure condition (assuming that the facet gradients are bounded from below across the family of colexes of interest). In the first example a small value of $\overline{d_\kappa}(v)$ saves the day, whereas in the second case $\overline d(v)$ makes the difference for some vertices, as the one depicted. 
}
\label{fig:forbidden}
\end{figure}

\begin{lem}\label{lem:closure}
If the cells of the tetrahedral colex have at most $k_\text{face}$ faces and
\begin{itemize}
\item
for every cell $c\not\in C_i$ the set of points of $c$ that also belong to some element of $C_i$ is simply connected, and
\item
for any color $\kappa$ and any vertices $v,v'$ of $\Gamma_i$
\begin{equation}
\overline d(v,v')\leq k \,d(v,v'),
\end{equation}
\begin{equation}
\min\left(\overline{d_\kappa}(v), \overline d(v)\right)\leq k\, d_\kappa(v),
\end{equation}
\end{itemize}
then the closure condition is satisfied for
\begin{equation}
k_\text{close}=4k(k_\text{face}-1).
\end{equation}
\end{lem}
\noindent
The proof is in appendix~\ref{sec:flux}.
Neither the first condition nor the first inequality are a big source of difficulties. Assuming that the lattice is regular enough, the first inequality holds for some $k$ (independent of the code size) as long as the $\Phi_i$-to-$\overline{\Phi_i}$ interfaces are \emph{simply connected and \emph{approximately} flat and convex}. The second inequality, however, forbids certain overall geometries. This is illustrated in figure~\ref{fig:forbidden}.

\section{Twister architecture}\label{sec:twister}

\begin{figure}
\centering
\includegraphics[width=\columnwidth]{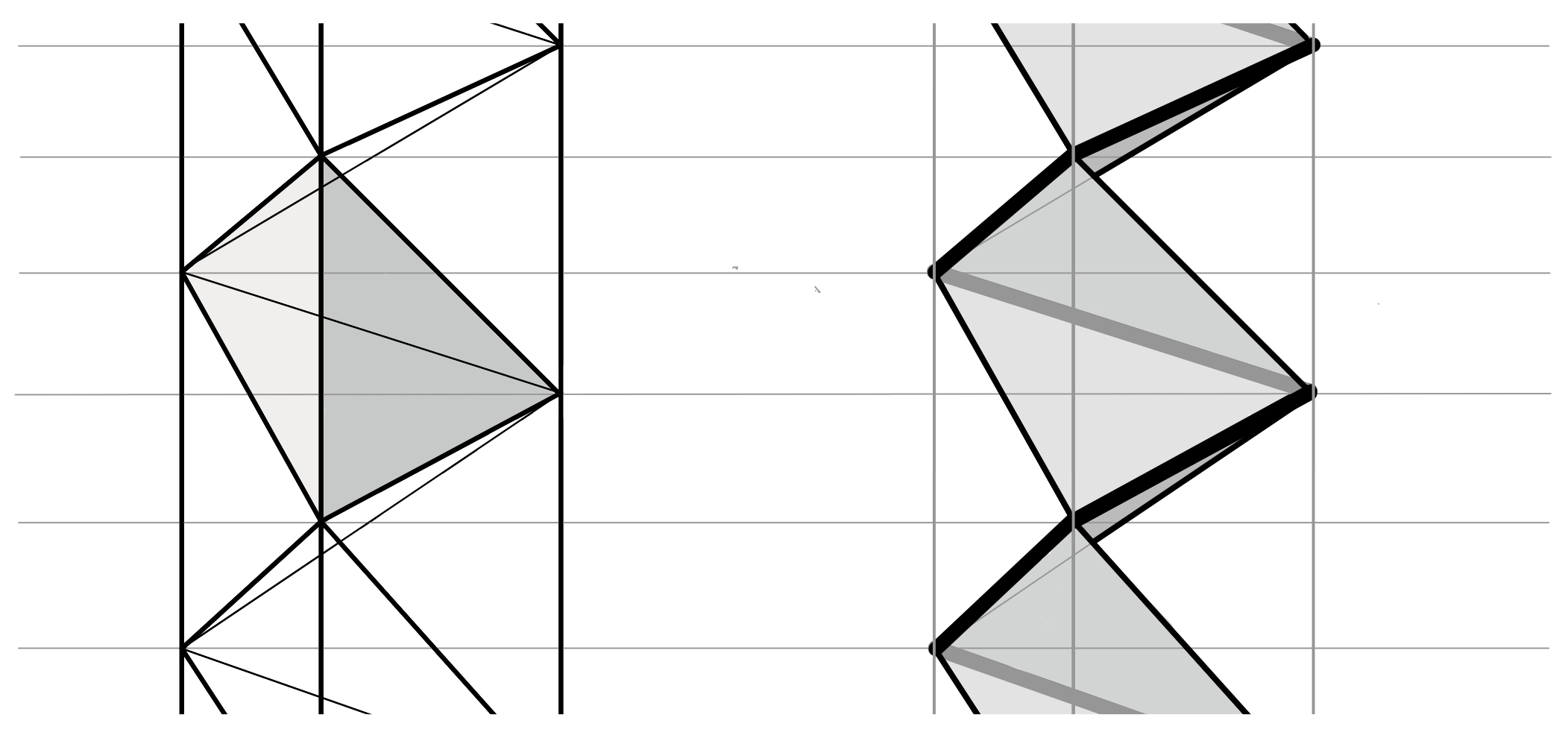}
\caption{
Lateral projection of the `twister' arrangement of tetrahedra on a prism. The tetrahedra are related by a discrete helicoidal symmetry. The horizontal lines are a guide to the eye. 
(Left) A single tetrahedron is highlighted.
(Right) The inner facets, those separating the tetrahedra, are shaded, and the edges along which these facets meet are highlighted.
}
\label{fig:linear}
\end{figure}

This section introduces a symmetrically arrangement of tetrahedra compatible with the geometrical conditions of section~\ref{sec:spacetime}.

\subsection{Linear logical graph}

Linear graph states enable MBQC for a single computational logical qubit. As shown in figure~\ref{fig:linear}, the corresponding tetrahedra can be arranged on a `twister' geometry: filling up a prism according to a discrete helicoidal symmetry. For the hybrid and two-dimensional schemes the time direction is the prism's axis, \emph{i.e.} the direction along which the logical computation proceeds. Notice that the spacetime geometry of each tetrahedron corresponds to one of the cases in figure~\ref{fig:forbidden}.

\subsection{Multi-qubit computation}

\begin{figure}[t]
\centering
\includegraphics[width=.8\columnwidth]{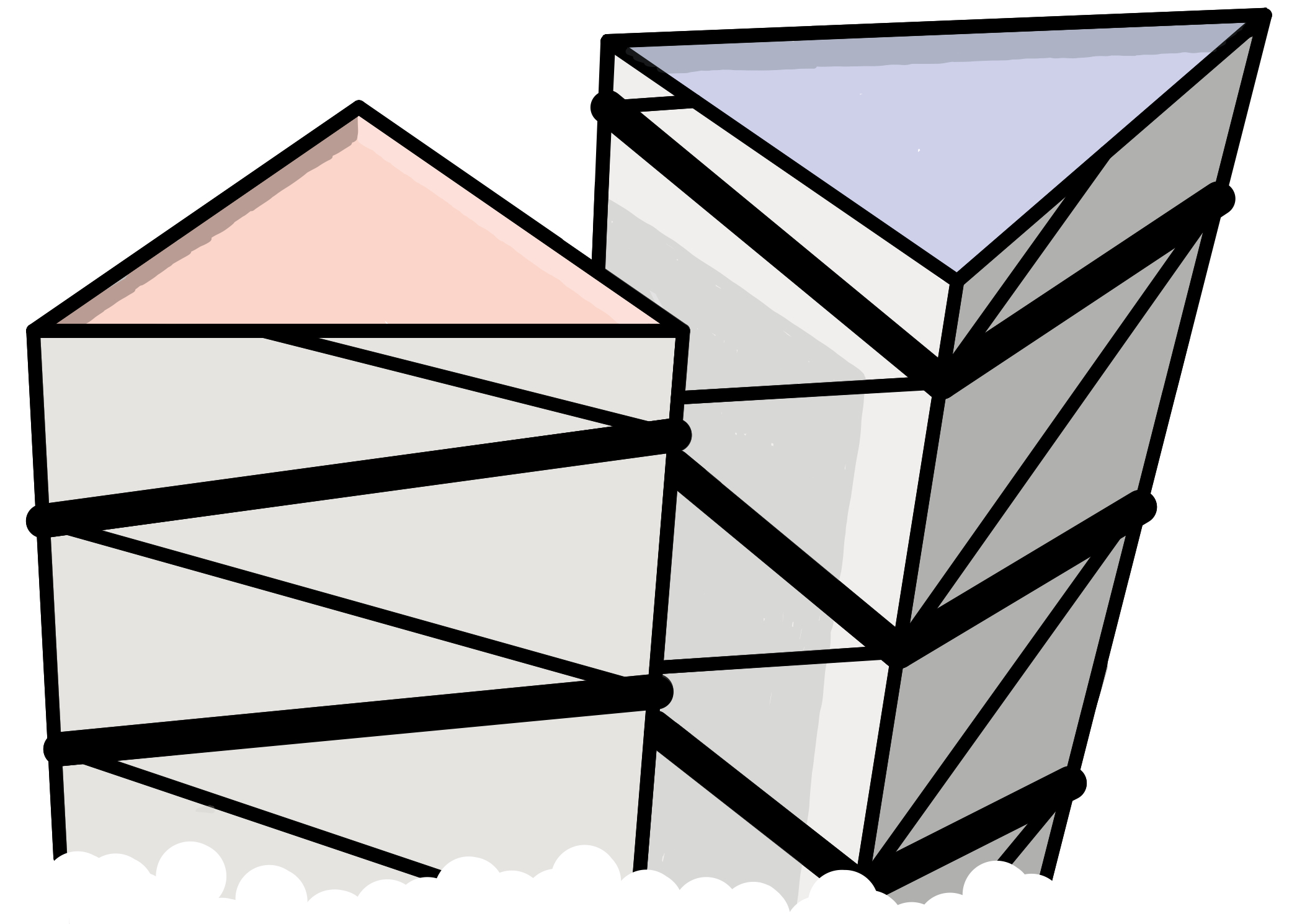}
\caption{Two `twister' configurations of tetrahedra with opposite chirality. Each contributes a logical computational qubit. When placed side by side the facets can be matched, enabling 2-qubit logical gates. Some edges are highlighted to make the symmetry more apparent.}
\label{fig:twisters}
\end{figure}

As suggested in figure~\ref{fig:twisters}, two `twister' arrangements of tetrahedra can be placed side by side as long as they have opposite chirality. They can be arranged in a row or fill the whole space, see figure~\ref{fig:pattern}. In the former case the resulting logical graph state is planar, see figure~\ref{fig:planar}. From the circuit model perspective this enables computation on a logical one-dimensional array. In the latter case the resulting graph is three-dimensional, which corresponds to computation on a two-dimensional array of logical qubits. In both cases the topology of the logical graph can be modified by removing some edges, \emph{i.e.} by not matching the corresponding facets in the construction of the resource state.

\begin{figure}[t]
\centering
\includegraphics[width=\columnwidth]{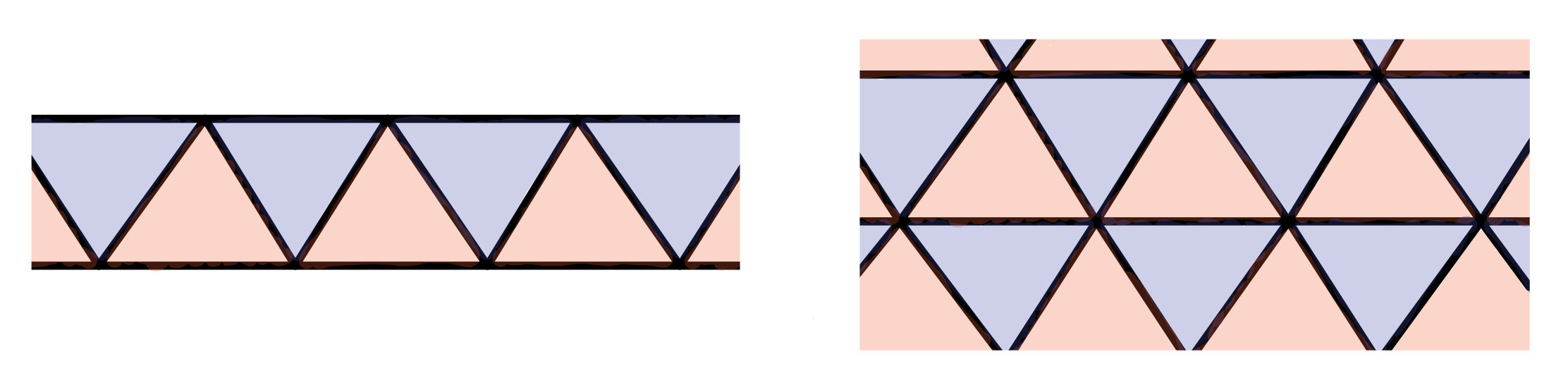}
\caption{Possible arrangements of `twisters'. Each twister contributes a logical computational qubit, and thus the pattern also reflects the topology of the logical circuit. (Left) A single row. (Right) A space-filling pattern. }
\label{fig:pattern}
\end{figure}

\begin{figure}[!t]
\centering
\includegraphics[width=.6\columnwidth]{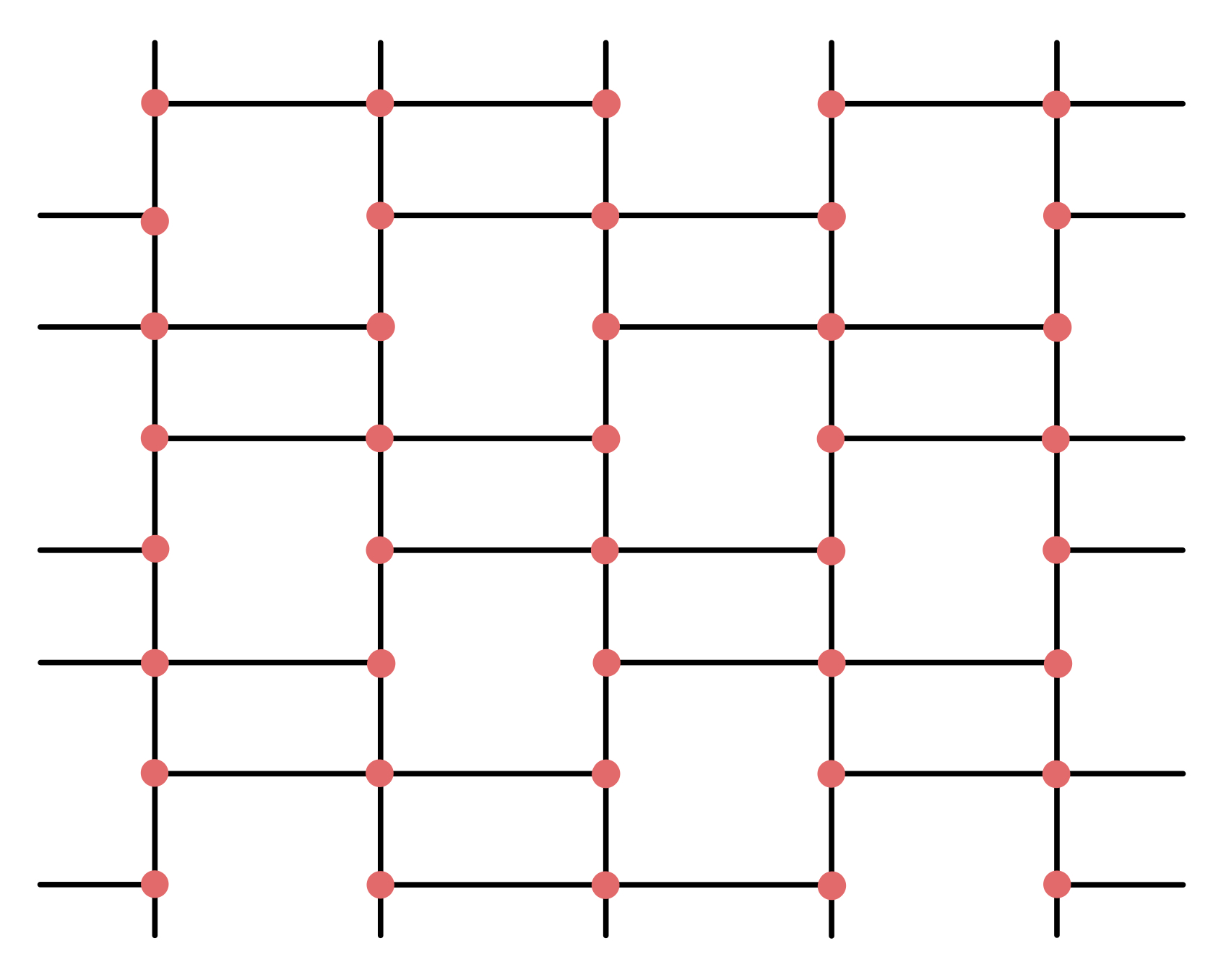}
\caption{The logical graph resulting of the arrangement of figure~\ref{fig:pattern} (left).}
\label{fig:planar}
\end{figure}

The three dimensional logical graph resulting from the space-filling pattern of figure~\ref{fig:pattern} (right) is dual to a 3-colex%
\footnote{
In particular, the 3-colex with the unit cell shown in figure~\ref{fig:unit}, which is discussed in \cite{bombin:2015:gauge}.
}.
Thus, the twister architecture enables to have a single symmetry govern both the overall geometry of the tetrahedral colexes and the geometry of the underlying 3-colex lattice.

\section{Naive JIT decoder}\label{sec:naive}

This section introduces a JIT decoder that, at each step and using conventional decoders,
\begin{itemize}
\item
estimates a portion of the $X$-error syndrome with the available information, and
\item
compensates for the difference between the current estimate and the one made in the previous step.
\end{itemize}
This `naive' approach is fault tolerant.

\subsection{Simplification}

We adopt the following simplifying assumption for decoding processes.

\begin{danger}
\center Classical computation and communication are instantaneous.
\end{danger}

\noindent In the limit of large systems, this is unphysical. The time required for decoding grows with the code size, whereas the available time is by assumption independent of the code size. Similarly, as the code size grows information becomes farther spread in space, but information can only travel at a finite speed.

We only consider decoding tasks for which there exist efficient algorithms%
\footnote{This avoids any possible logical pitfalls, since there is little point to have a quantum computer if classical computation is instantaneous.}. 
Moreover, it is likely that \emph{these tasks can also be successfully accomplished under tight physical assumptions}. Indeed, both renormalization group~\cite{duclos:2010:fast, sarvepalli:2012:efficient, duclos:2013:fault} and cellular automaton~\cite{herold:2015:cellular, herold:2017:cellular, dauphinais:2017:fault, kubica:2018:cellular} ideas have already been applied in very similar scenarios.

\subsection{Decoding problem}

Consider the decoding problem posed by the linear code $\mathcal C$ formed by $X$-error syndromes%
\footnote{
This problem first appeared in~\cite{bombin:2015:single-shot}: it is the first stage of single-shot error correction for 3D gauge color codes, where a noisy version of the so called `gauge syndrome' (either for $X$ or $Z$ errors) is processed to obtain a corrected gauge syndrome and, from it, a syndrome. It is also the first stage of single-shot error correction for $X$ errors in 3D color codes~\cite{bombin:2018:transversal}, which is the problem that the JIT decoder addresses.
}.
A decoder $D$ for this code is a map
\begin{equation}
D:\mathcal P(\Phi)\longrightarrow\mathcal C.
\end{equation}
where $\Phi$ is as in section~\ref{sec:2D_setting}. It satisfies for any $\phi\in\mathcal C$ and any $\omega\subseteq \Phi$,
\begin{equation}
D (\omega+\phi) =
D(\omega)+\phi.
\end{equation}
\emph{I.e.} the decoder estimates an error based solely on the sydrome of its input. The decoders considered below are assumed to satisfy analogous conditions.

\subsubsection{$\mathbf Z_2$ charge decoding.}

The decoding problem above is a $\bf Z_2^3$ analogue of the $\bf Z_2$ problem faced in the fault-tolerant correction of $X$ (or $Z$) errors in the 2D toric code (with an important difference to be noted below). In particular, it is possible to map the former problem to the later at the cost of a constant factor in the weight of the decoded errors~\cite{bombin:2015:single-shot}. Since the $\mathbf Z_2$ case is simpler and well known, we use it to illustrate the naive JIT decoder. 

The $\mathbf Z_2$ decoding problem is as follows~\cite{dennis:2002:tqm}, see also appendix~\ref{sec:z2_decoding}. There is a 3D lattice with set of edges $E$. The code is formed by closed sets ($\mathbf Z_2$ chains) of edges, and the inputs for decoding are arbitrary sets of edges $E$. In the case of fault-tolerant error correction in the toric code, only the homology class of the decoded set of edges is relevant, whereas for the JIT decoder it is the actual codeword that matters. The difference is, however, immaterial if the decoder approximates the most likely error (as in minimum weight matching~\cite{dennis:2002:tqm}), rather than the most likely equivalence class.

\begin{figure}
\centering
\includegraphics[width=.6\columnwidth]{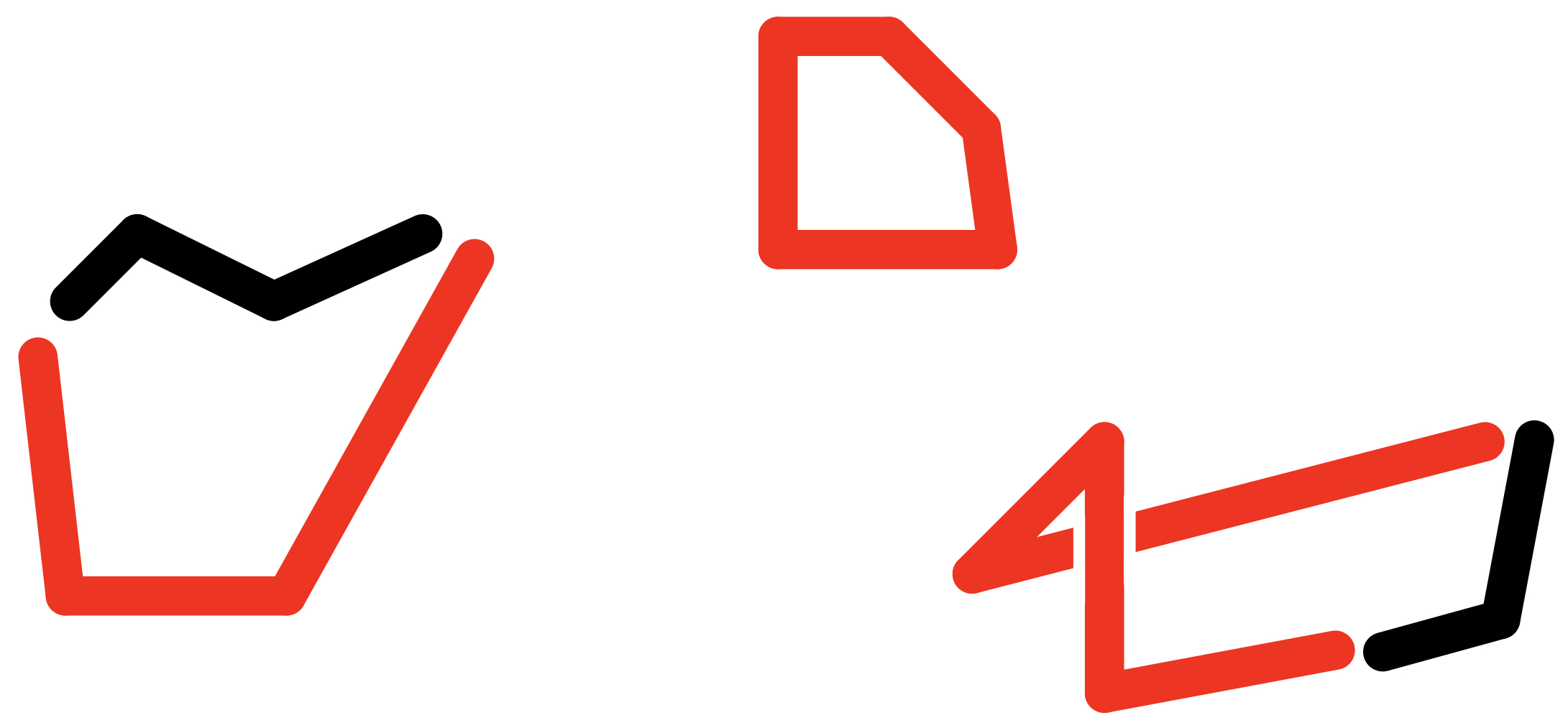}
\caption{
The goal of the $\mathbf Z_2$ charge decoding problem is to find a likely error (black set of edges) with the same endpoints as the actual error (red).
}
\label{fig:z2}
\end{figure}

The errors involved in the process are as follows. The input is some subset of edges $\tilde e\subseteq E$. Suppose that it takes the form $e+\omega$, where $e$ is a codeword and $\omega$ the error. The error syndrome is extracted from $\tilde e$: it is the set of endpoints of $\tilde e$, which coincide with the endpoints of $\omega$. Decoding amounts to choose some $\omega'$ with the same endpoints as $\omega$, see figure~\ref{fig:z2}:
\begin{equation}
e +\omega 
\quad\overset D\longmapsto \quad 
e + (\omega+\omega').
\end{equation}
The residual error is $\omega+\omega'$.

\subsection{Adapted decoders}

The naive JIT decoder requires two modified versions of the decoder $D$. Recall the layer structure introduced in section~\ref{sec:2D_setting}. For each layer $i$ we consider 
\begin{itemize}
\item
a decoder $D_i$ for errors in $\Phi_i$ with \emph{open} boundary conditions, used for the $i$-th \emph{estimating} step, and 
\item
a decoder $\overline {D_i}$, based on a decoder for errors in $\overline{\Phi_i}$ with \emph{closed} boundary conditions, and used for the $(i+1)$-th \emph{compensating} step. 
\end{itemize}
`Boundary' here refers to the interface between $\Phi_i$ and $\overline{\Phi_i}$, which in practice takes the form of a 2D sheet separating the first $i$ layers from the rest.

\begin{figure}
\centering
\includegraphics[width=\columnwidth]{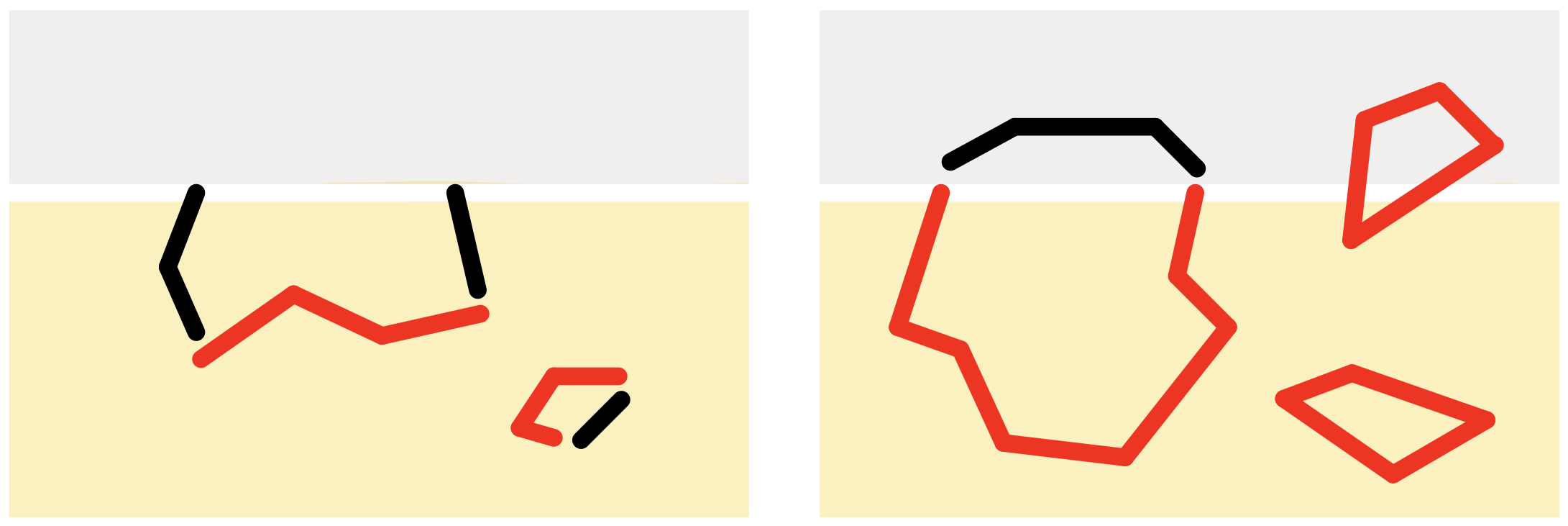}
\caption{
The workings of the decoders $D_i$ (left) and $\bar D_i$ (right). The orange and grey regions represent $\Phi_i$ and $\overline{\Phi_i}$ respectively.
The sets of edges $\omega$ and $\omega'$ are represented in red and black, respectively, see equations~\eqref{eq:open_error} and~\eqref{eq:closure_error}.
}
\label{fig:decoders}
\end{figure}

\begin{warning}
The {\bf open boundary conditions decoder} is a function
\begin{equation}
D_i:\mathcal P(\Phi_i)\longrightarrow\mathcal C_i.
\end{equation}
\end{warning}
\noindent 
For this decoder the original decoding problem is modified by imposing that for any $\omega\subseteq \Phi_i$ and $\gamma\subseteq \overline{\Phi_i}$, the error $\omega+\gamma$ is as likely as the error $\omega$. Thus it is meaningful to have subsets of $\Phi_i$ as both input and output: the $\overline{\Phi_i}$ component is completely random and irrelevant. It is in this sense that the boundary conditions are open, see appendix~\ref{sec:open}.
Decoding takes the form
\begin{equation}\label{eq:open_error}
\phi\cap\Phi_i + \omega
\quad\overset{D_i}\longmapsto \quad 
\phi\cap\Phi_i+\omega+\omega',
\end{equation}
where the input is composed by some $\phi\in\mathcal C$ and an error $\omega\in\Phi_i$. In the $\mathbf Z_2$ picture, illustrated in figure~\ref{fig:decoders} (left), $\omega$ and $\omega'$ share endpoints \emph{except} possibly at the interface of $\Phi_i$ and $\overline{\Phi_i}$. .

\begin{warning}
The {\bf closed boundary conditions decoder} is a function
\begin{equation}
D_i':\mathcal P\left(\overline{\Phi_i}\right)\longrightarrow\mathcal C\cap \mathcal P\left(\overline{\Phi_i}\right).
\end{equation}
\end{warning}
\noindent
For this decoder the original decoding problem is constrained to the subset of edges $\overline{\Phi_i}$: these are the only noisy ones. Rather than using this decoder directly, it will be used in the `compensating' step, as indicated above. To this end, let us recast it as
\begin{equation}
D_i'(\phi)=: \phi + E_i(\phi),
\end{equation}
where $E_i$ is a function that \emph{only depends on the syndrome of its argument} and yields a Pauli operator with that syndrome, \emph{i.e.} $E_i(\phi)$ is the estimated error for $\phi$.

\begin{warning}
The {\bf closure decoder} is the function
\begin{align}
\overline {D_i}:\mathcal C_i' &\longrightarrow\mathcal C,
\\
 \phi &\longmapsto \phi + E_i(\phi),
\end{align}
with domain
\begin{equation}
\mathcal C_i' := \{
\phi+\omega
\,|\,
\phi\in \mathcal C,\,\, \omega\in \mathcal C_i
\}.
\end{equation}
\end{warning}
\noindent 
Notice that only the $\overline {\Phi_i}$ component of $\phi$ is modified, \emph{i.e.} for every $\phi\in\mathcal C_i'$
\begin{equation}
\phi+ \overline {D_i}(\phi)\,\in\overline{\Phi_i}.
\end{equation}
The $\Phi_i$ component of the input sets the boundary conditions to be satisfied by the $\overline{\Phi_i}$ component of the output, which `closes the open ends'.
Decoding takes the form
\begin{equation}\label{eq:closure_error}
\phi + \omega
\quad\overset{\overline{D_i}}\longmapsto \quad 
\phi+(\omega+\omega'),
\end{equation}
where $\phi\in \mathcal C$, $\omega\in\mathcal C_i$ and $\omega'\in\overline{\Phi_i}$. In the $\mathbf Z_2$ picture, illustrated in figure~\ref{fig:decoders} (right),  $\omega$ and $\omega'$ share endpoints.

Appendix~\ref{sec:open} provides a simple criteria guaranteeing that `open' boundary conditions are truly so. For reasonable partitions of the lattice into layers it should in general be straightforward to adapt any decoding algorithms $D$ in order to obtain algorithms $D_i$ and $\overline{D_i}$. Indeed, the situation is not conceptually different from the well known `smooth' and `rough' boundaries in toric codes.

\subsection{Computation}

At the $i$-th layer the naive JIT decoder produces a syndrome $\hat \phi$ of the form
\begin{equation}
\hat \phi \cap \Lambda_i := \gamma_i\cap \Lambda_i,
\end{equation}
where $\gamma_i$ is computed, as stated above, in two steps, estimation 
\begin{equation}\label{eq:estimate} 
\gamma_i' :=D_i(\tilde\phi\cap \Phi_i) 
\end{equation}
and compensation
\begin{equation}\label{eq:compensate}
\gamma_i := \gamma_i' + \overline{D_{i-1}}(\gamma_{i-1} + \gamma_i'\cap \Phi_{i-1}).
\end{equation}
with $\gamma_0 := \emptyset$%
\footnote{
A more natural expression would be
\begin{equation}\label{eq:compensate2}
\gamma_i := \gamma_i' + \overline{D_{i-1}}\left(\left(\gamma_{i-1} + \gamma_i'\right)\cap \Phi_{i-1}\right).
\end{equation}
We use~\eqref{eq:compensate} so that equation~\eqref{eq:JIT_epsilon2} holds, which is used in proving lemma~\ref{lem:naive}.
Both definitions are equivalent as long as $\gamma_i$ is always inside $\Phi_i$, but this needs not be the case in general.
}.
The above is well defined because $\gamma_i\in\mathcal C_i'$. It is not difficult to check that
\begin{equation}
\hat\phi\cap\Phi_i = \gamma_i\cap\Phi_i.
\end{equation}
It is worth pointing out that the same decoders can be used when information from farther layers is available, by settling on the values of $\hat \phi$ for a single layer but computing with all the available layers.

\subsection{Errors}\label{sec:errors}

\begin{figure}
\centering
\includegraphics[width=\columnwidth]{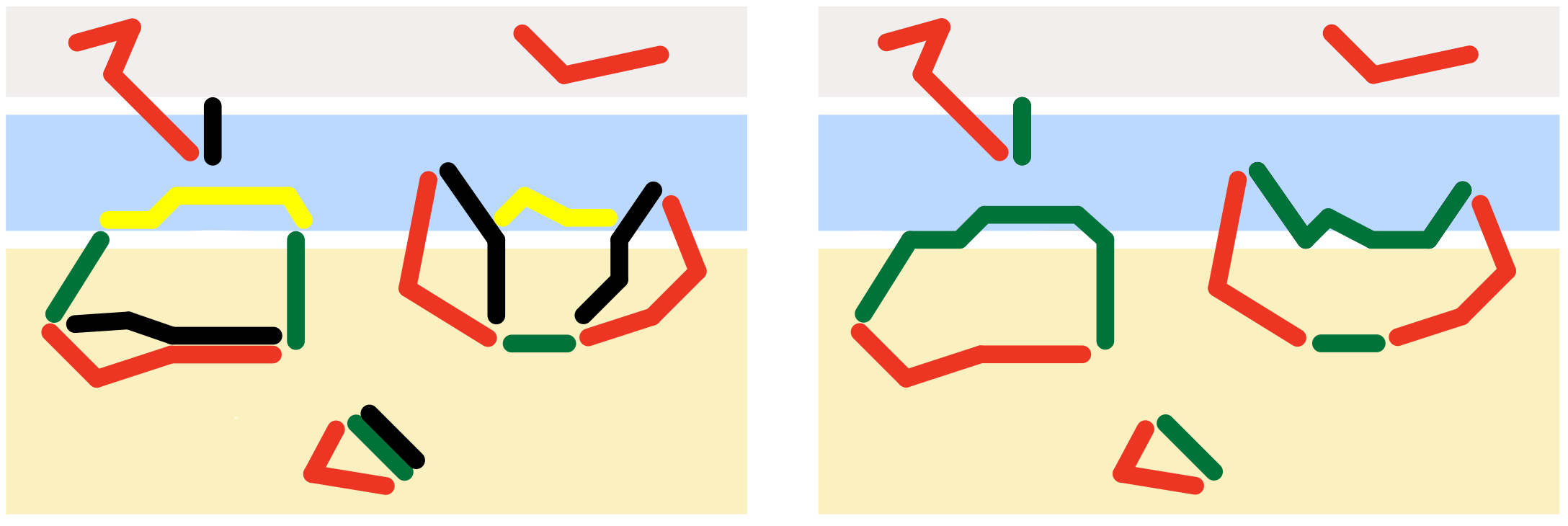}
\caption{
The workings of the naive JIT decoder. The orange, blue and grey region represent $\Phi_{i-1}$, $\Lambda_i$ and $\overline{\Phi_i}$, respectively.
The sets of edges $\omega$ is represented in red, $\omega_i'$ in black, $\epsilon_i$ in yellow, $\omega_{i-1}$ in green (left only) and $\omega_i$ in green (right only). See equations~(\ref{eq:JIT_omega}-\ref{eq:JIT_epsilon}).
}
\label{fig:JIT_decoder}
\end{figure}

Let us rephrase the computation in terms of several different kinds of `errors'. 

\begin{warning}
Let $\phi\in \mathcal C$ be the noiseless flux configuration (a syndrome) and $\tilde \phi\subseteq \Phi$ the noisy one. We define

\begin{itemize}
\item
the error $\omega$ in the measurements
\begin{equation}\label{eq:JIT_omega}
\omega := \phi + \tilde\phi,
\end{equation}
\item
the residual error $\hat\omega$  of JIT decoding
\begin{equation}
\hat\omega := \phi + \hat\phi,
\end{equation}
\item
the estimated error $\omega_i'$ at the $i$-th step
\begin{equation}
\omega_i' := \gamma'_i + \tilde\phi\cap\Phi_i,
\end{equation}
\item
the effective estimated error $\omega_i$ at the $i$-th step
\begin{equation}
\omega_i := \gamma_i + \tilde\phi\cap\Phi_i,
\end{equation}
\item
the compensating flux configuration $\epsilon_i$ at the $i$-th step
\begin{equation}\label{eq:JIT_epsilon}
\epsilon_i := E_i\left(\gamma_{i-1} + \gamma_i'\cap \Phi_{i-1}\right).
\end{equation}\end{itemize}
\end{warning}
\noindent
A crucial property is that, since $\gamma_i$ and $\gamma_i'$ have the same syndrome, $\epsilon_i$ only depends on the estimated errors of two consecutive guessing steps:
\begin{equation}\label{eq:JIT_epsilon2}
\epsilon_i =  E_i\left(\omega_{i-1}' + \omega_i'\cap \Phi_{i-1}\right).
\end{equation}
The corresponding $\mathbf Z_2$ picture is illustrated in figure~\ref{fig:JIT_decoder}: $\omega_i'$ (and thus $\omega_i$) has the same endpoints as $\omega\cap\Phi_i$, except at the interface separating $\Phi_i$ and $\overline{\Phi_i}$.
The expressions~\ref{eq:estimate} and~\ref{eq:compensate} can be restated in terms of these various errors,
\begin{align}
\omega_i' &= \omega\cap\Phi_i + D_i(\omega\cap\Phi_i),
\\
\omega_i &= \omega_{i-1} + \omega_i'\cap \Lambda_i + \epsilon_i,
\end{align}
and the residual error of JIT decoding is
\begin{equation}
\hat \omega = \omega_n = 
\sum_i \left(
\omega_i'\cap\Lambda_i + \epsilon_i
\right).
\label{eq:JIT_error}
\end{equation}

\subsection{Fault tolerance}\label{sec:ft_2D}

Under the same assumptions as in section~\ref{sec:ft}, the naive JIT decoder makes the 2D scheme fault tolerant. As discussed next, this is true in particular
\begin{itemize}
\item
for decoders with an efficient implementation, and
\item
in an unfavorable scenario that omits the strategy of section~\ref{sec:delta}.
\end{itemize}

Just as in the three-dimensional scenario, the preparation of the encoded graph state is the crux of the scheme's fault-tolerance. In this case there is an additional complication: the residual error of JIT decoding~\eqref{eq:JIT_error} does not seem amenable to the `confined syndromes' approach of~\cite{bombin:2015:single-shot, bombin:2018:transversal}. 

\begin{figure}
\centering
\includegraphics[width=.6\columnwidth]{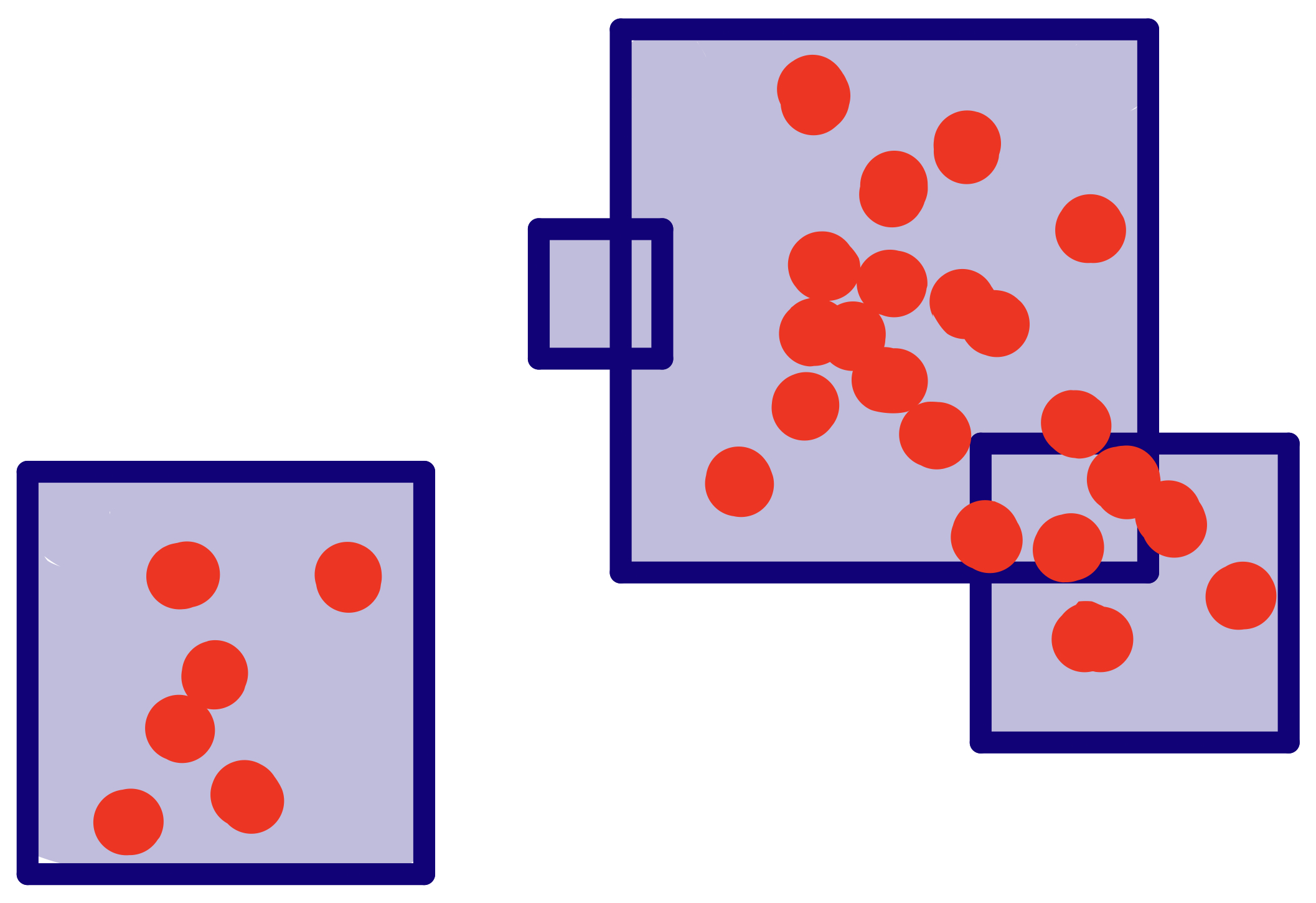}
\caption{
Errors (red dots) confined to the interior of a collection of balls.
}
\label{fig:ball-local}
\end{figure}

The workaround is to consider an alternative to local noise that is well-suited to topological codes%
\footnote{This is briefly motivated in appendix~\ref{sec:locality}.
}.
The key element are geometrical balls, defined by some metric that reflects the structure of errors and syndromes. It is assumed that that each error $e$ is mapped to a ball set $W$ such that the support of $e$ is contained within the balls in $W$, see figure~\ref{fig:ball-local}. Ball-local noise is defined in terms of the distribution of balls $W$.

\begin{warning}
A distribution of error operators, each with support constrained to the interior of some set of balls $W$, is {\bf ball-local} with rate $0\leq p\leq 1$ if for any set of balls $B$
\begin{equation}\label{eq:ball-local}
\text{prob}(B\subseteq W)\leq p^{r(B)},
\end{equation}
where $r(B)$ is the sum of the radii of the balls in $B$.
\end{warning}

\noindent
Notice that $\hat\omega$ is the JIT counterpart to the error $\bar\omega$ of section~\ref{sec:ft}.
By the same argument used for $\bar\omega$, the $X$ component of the effective Pauli frame error $F_{\hat\omega}$ is only fixed up to a logical error. The following result, technically stated in theorem~\ref{thm}, assumes that (see appendix~\ref{sec:naive} for details)
\begin{itemize}
\item
the closure condition holds, 
\item
the family of tetrahedral colexes has a uniform local structure (such as the one described in~\cite{bombin:2015:gauge}), and 
\item
the decoders $D_i$ and $D_i'$ satisfy certain technical conditions (satisfied for the decoder $D$ described in~\cite{bombin:2015:single-shot}). 
\end{itemize}

\begin{success}
If the error rate for ancillas outcomes is below a threshold $p_0$, the $X$ component of $F_{\hat\omega}$ follows a ball-local distribution with rate
\begin{equation}
\left( 
\frac p{p_0}
\right)^k,
\end{equation}
where both $p_0$ and $k$ depend on the local structure of the colexes and on the decoders $D_i$, $D_i'$.
\end{success}

\noindent
When the $X$ component of $F_{\hat\omega}$ follows a ball-local distribution, so does $F_{\hat\omega}$ too, because it is obtained from its $X$ component via (local) logical CP gates.
Observe that the code qubits measurements affected by $F_{\hat\omega}$ do not belong to logical qubits measured in the $Z$ basis.
Therefore, the logical outcome processing affected by $F_{\hat\omega}$ always involves a $\mathbf Z_2^3$ charge decoder%
\footnote{Logical measurements in the $X$, $Y$ or $X\pm Y$ bases are all equivalent to a transversal gate~\cite{bombin:2015:gauge} followed by a transversal $X$ measurement. The relevant syndrome for the decoding of $X$ measurements is that of cell operators ($X$ stabilizer generators). This syndrome can be understood in terms of $Z_2^3$ charges~\cite{bombin:2018:transversal}. Notice that, despite sharing the same nature, this decoding problem and the decoding problem of $X$-error syndromes are different.
}.
This decoding problem can be mapped to three connected copies of the $\mathbf Z_2$ charge problem~\cite{bombin:2015:single-shot, kubica:2015:unfolding}.
This mapping preserves the ball-local character of the error distribution%
\footnote{This comes at a cost: if the error rate is $p$ then via the mapping it becomes $p^{1/3}$: a single ball in the original geometry gives rise to three different balls on the mapped geometry.
}.
Appendix~\ref{sec:z2} shows that, for the $\mathbf Z_2$ charge problem, there exist efficient decoders with an error correction threshold for ball-local noise. Thus, fault-tolerance is indeed feasible.

\section{Discussion}

Two-dimensional colorful quantum computation exemplifies the importance of investigating fault-tolerance with a focus on processes, rather than error-correcting codes per se. Given the prevalence of stabilizer-based techniques in the field, graph states offer an excellent tool to adopt such a point of view. This is perhaps particularly apparent in the realm of topological methods, where the graph state picture reveals otherwise hidden symmetries and simplifies the description of protocols.

The analysis of two-dimensional colorful quantum computation above sticks to the graph state picture. This removes the need to understand the states of the system as the computation proceeds, or to construct explicit circuit models. However, these are aspects that deserve to be studied. More generally, it would be desirable to understand colorful quantum computation from a condensed matter / TQFT perspective.

As presented, the fault tolerance of two-dimensional colorful quantum computation relies on unphysical assumptions, because the time required for classical computation and communication in JIT decoding is completely ignored. Hopefully these assumptions can be lifted, if not theoretically, at least as part of a more practical exploration of JIT decoders.

The current knowledge of error thresholds and decoders for 3D color codes is limited~\cite{kubica:2018:three, brown:2016:fault}. It is unclear what the impact of JIT decoding might be on error thresholds, in particular when compared to the conventional three-dimensional scenario. A reason for optimism is that JIT decoding mostly contributes erasure-like errors, which have a much more benign impact than other forms of noise. Moreover, adding limited amounts of delay opens up the possibility of interpolating between the most extreme two-dimensional case (whatever it is) and the three-dimensional scenario.

\vspace{\baselineskip}
\noindent
{\bf Acknowledgements.}  The bulk of the resource state used in colorful quantum computation first appeared on discussions with Naomi Nickerson over alternatives transcending the foliated approach~\cite{bolt:2016:foliated} to MBQC with 2D color codes, somewhat along the lines of~\cite{nickerson:2018:measurement}. I  would  like  to  thank  the  whole  PsiQuantum  fault  tolerance  team, Christopher  Dawson,  Fernando  Pastawski,  Kiran  Mathew,  Naomi Nickerson,  Nicolas  Breuckmann, Andrew Doherty, Jordan Sullivan, and Mihir  Pant for  their  considerable  support and  encouragement. In particular I would like to thank Terry Rudolph,  Nicolas Breuckmann,  Naomi Nickerson,  Fernando Pastawski,  Mercedes Gimeno-Segovia,  Peter  Shadbolt  and  Daniel  Dries  for  many  useful  discussions and/or  very  generous  feedback  at  various  stages  of  this  manuscript.

\appendix

\section{Notation}\label{sec:notation}

\begin{warning}

\begin{itemize}
\item
The symmetric difference of sets is represented with $+$ (lower precedence than $\cup, \cap$).
\item
$\mathcal P(A)$ is the powerset of $A$.
\item
$P$ is \emph{the} Pauli group, and \emph{a} Pauli group is any of its subgroups. 
\item
$P_X$, $P_Z$ are the Pauli groups generated by $X$ and $Z$ operators respectively. 
\item
$a|_x\propto a$ if $a=b\otimes c$ is a Pauli operator on a system $x\otimes y$.
\item
$\partial a$ is the error syndrome of $a$.
\end{itemize}
 
Given Pauli groups $A, B$:
\begin{itemize}
\item 
$A|_x$ contains the elements of $A$ restricted to the subsystem $x$, i.e. it is the group $\{a|_x \,|\, a \in A\}$. Notice that $A|_x=\langle i \rangle\, A|_x$.
\item 
$A\|_x$ is the subgroup of elements of $A$ with support in the subsystem $x$, i.e. it is the set $\{a_x \,|\, a_x\otimes \mathbf 1_y \in A\}$.
\item
$\mathcal Z_A(B)$ is the subgroup of elements of $A$ that commute with the elements of $B$. We use shorthands such as $\mathcal Z_X(\cdot):= \mathcal Z_{P_X}(\cdot)$ or $\mathcal Z_r(\cdot):= \mathcal Z(\cdot)\|_r$.
\end{itemize}

\end{warning}

\section{Piecewise Pauli frame}\label{sec:piecewise}

The aim of this appendix is to verify that the causality condition of section~\ref{sec:causality} ensures that the Pauli frame can be obtained in a piecewise manner. We start with the following observation.

\begin{lem}

Given a Pauli group $A$ and a subset of qubits $r$
\begin{equation}\label{lem:dual}
\mathcal Z(A)|_r = \mathcal Z_{r}(A\|_r).
\end{equation}
\end{lem}

\begin{proof} Observe that, dually,
\begin{equation}
\mathcal Z(A)\|_r = \mathcal Z_{r}(A)= \mathcal Z_{r}(A|_r),
\end{equation}
so that
\begin{equation}
\mathcal Z(A)|_r = \mathcal Z_r(\mathcal Z_r(\mathcal Z(A)|_r)) = \mathcal Z_r(\mathcal Z(\mathcal Z(A))\|_r)
=\mathcal Z_{r}(A\|_r).
\qedhere
\end{equation}
\end{proof}

Assuming that the causality condition is satisfied, at the $i$-th step the synfrome $\phi_i$ of $S_Z\|_{R_i}$ is accesible. In particular, $\phi_i$ is the restriction to $S_Z\|_{R_i}$ of some syndrome $\phi$ of $S_Z$ (independent of $i$). Lemma~\ref{lem:prescription} below guarantees the feasibility of the following straightforward prescription to choose the Pauli frame at the $i$-th step.
\begin{success}
Choose $q_i\in P_X$ with support in $R_i$, syndrome $\phi_i$ over $S_Z\|_{R_i}$ and such that
\begin{equation}
q_i|_{R_{i-1}}\propto q_{i-1}.
\end{equation}
\end{success}

\begin{lem}
\label{lem:prescription}
If $q_i\in P_X$ has support on $R_i$ and syndrome $\phi_i$ over $S_Z\|_{R_i}$, there exists $\bar q_i\in P_X$ with no support on $R_i$ and such that $q_i\bar q_i$ has syndrome $\phi$ over $S_Z$.
\end{lem}
\begin{proof}
Choose any $q\in P_X$ with syndrome $\phi$ over $S_Z$. It satisfies
\begin{equation}
q|_{r_i} q_i\in \mathcal Z_{r_i}(S_Z\|_{r_i})=\mathcal Z(S_Z)|_{r_i},
\end{equation}
where the equality is by lemma~\ref{lem:dual}.
Then there exists some $o\in \mathcal Z_X(S_Z)$ with
\begin{equation}
q|_{r_i}q_i=o|_{r_i}.
\end{equation}
It suffices to take
\begin{equation}
\bar q_i:= q q_i o
\end{equation}
because
\begin{equation}
\bar q_i|_{r_i}\propto \mathbf 1,
\end{equation}
and $q_i\bar q_i=qo$ has syndrome $\phi$ over $S_Z$.
\end{proof}

\section{Ball colexes}\label{sec:ball}

The aim of this appendix is to prove lemma~\ref{lem:causality}. Throughout the section, $S_XS_Z$ is the 3D color code stabilizer group for a generic 3-colex $K$ (with no a priori constraints on its geometry or topology), and $R_X$ is the stabilizer group generated by its facet operators $X_r$. When there is another colex $K'$ or $\bar K$, we consider analogous primed or barred symbols (\emph {e.g.} $S_X'$). For the definitions of 3-colexes and 3D color codes, see~\cite{bombin:2018:transversal}.

\begin{warning}
A 3-colex is a {\bf ball} if
\begin{itemize}
\item
it is homeomorphic to a ball, and
\item
each of its facets is homeomorphic to a disc.
\end{itemize}
\end{warning}
\noindent
The following result is adapted from~\cite{bombin:2015:single-shot}. 
\begin{lem}\label{lem:ball}
For any ball 3-colex $K$
\begin{equation}
\mathcal Z_X (S_Z) \propto S_X\,R_X.
\end{equation}
\end{lem}

\begin{proof}[Sketch of proof] The key is the following duplication trick of~\cite{bombin:2015:single-shot}. Since $K$ is a ball it can be glued to a duplicate $K'$ of itself to produce a closed 3-colex $\bar  K$ with the topology of a 3-sphere. In particular, corresponding corners (vertices at which 3 facets meet) are joined with an edge, corresponding borders (connected components of a given intersection of two facets) are joined with a face and corresponding facets are joined with a cell.

Notice that
\begin{equation}
\bar S_Z = S_Z\cdot S'_Z \cdot \langle Z_f\rangle_{f\in F_I},
\end{equation}
where $F_I$ are the faces in the interface between the two copies (faces with vertices on both $K$ and $K'$). Since a 3D color code on a sphere has not logical qubits,
\begin{equation}
\bar S_X \propto \mathcal Z_X(\bar S_Z).
\end{equation}
Given an $X$ Pauli operator on the first copy $K$
\begin{equation}
p\in P_X|_K,
\end{equation}
let $p'$ be the analogous operator acting on the second copy $K'$. Interface faces are completely symmetrical up to the exchange of the copies, which implies that the product $pp'$ commutes with $Z_f$ for any interface face $f$. It follows that
\begin{align}
p\in \mathcal Z_X(S_Z)
&\iff pp'\in\mathcal Z_X(\bar S_Z) 
\nonumber\\
&\iff pp'\in \langle i \rangle \bar S_X \iff p\in \bar S_X|_ K.
\end{align}
Each cell $c$ of the interface conrresponds to a facet $r$ by construction, with
\begin{equation}
X_c|_K \propto X_r,
\end{equation}
which yields the desired result.
\end{proof}

Let $K$ be a 3-colex and $C$ its set of 3-cells. Any subset of 3-cells $C'\subseteq C$ induces a new 3-colex $K'$ in a natural way, \emph{i.e.} its cells are either elements of $C$ or subcells of those. Each facet $r$ of $K'$ corresponds to either a cell in $C'-C$ or a facet of $K$, the only one with faces in $r$. Lets call this cell/facet $w(r)$. Notice that the functions $w$ might not be injective.

\begin{warning}
Such a subcolex $K'$ is {\bf a ball of $K$} if
\begin{itemize}
\item
$K'$ is a ball, and
\item
the function $w$ is one-to-one.
\end{itemize}
\end{warning}
\noindent

\begin{lem}\label{lem:ball_of}

If the 3-colex $K$ is a ball and $K'$ a ball of $K$, then
\begin{equation}
S_Z\|_{ K'} = S'_Z.
\end{equation}
\end{lem}

\begin{proof}[Sketch of proof] 
Combining lemma~\ref{lem:ball} and the injectivity of the function $w$ we get
\begin{equation}
S'_XR'_X = (S_XR_X)|_{ K'}.
\end{equation}
Clearly $S_Z'\subseteq S_Z\|_{ K'}$, so it suffices to note that
\begin{align}
S_Z\|_{ K'}
&\propto \mathcal Z_Z(S_X R_X)\|_{ K'}
= \mathcal Z_{ K', Z}((S_X R_X)|_{ K'})
\nonumber\\
&= \mathcal Z_{ K', Z}(S'_X R_X')
\propto S_Z'.
\end{align}
\qedhere
\end{proof}

\begin{lem}[Restatement of lemma~\ref{lem:causality}]
The causality condition is satisfied if the subcolex with cells $C_i$ is a ball of the tetrahedral colex.
\end{lem}
\begin{proof}
Let $K'$ be the subcolex with cells $C_i$. By assumption the eigenvalues of the generators of $S_Z'$, \emph{i.e.} the face operators of the cells in $K'$, are known at the $i$-th step. The result follows by lemma~\ref{lem:ball_of}.
\end{proof}

\section{Flux and distance}\label{sec:flux}

The purpose of this section is to prove lemma~\ref{lem:closure}. Throughout this section it is assumed that some 3-colex is given such that any of its cells has at most $k_\text{face}$ faces.
For the definitions of 3-colexes and 3D color codes, see~\cite{bombin:2018:transversal}. We adopt the flux configuration picture for sets of faces $\phi$ and identify faces with their dual edges~\cite{bombin:2018:transversal}.

\subsection{Dual graph}

We need an extension of the dual graph $\Gamma$ defined in~\cite{bombin:2018:transversal}.
\begin{warning}
The {\bf extended dual graph $\hat\Gamma$} has
\begin{itemize}
\item
inner vertices, one per cell of the colex,
\item
outer vertices, one per facet of the colex,
\item
an edge for each face $f$ of the colex, connecting the vertices dual to the cells, or cell and facet, that $f$ is part of, and
\item
an edge for each border between facets, connecting the corresponding outer vertices.
\end{itemize}
The vertices of $\hat \Gamma$ are colored as their dual cells and facets.
\end{warning}

\subsection{Monopole configurations}

Consider the flux group $M$ defined in~\cite{bombin:2018:transversal}. For each color $\kappa$, let $M_\kappa$ be the subgroup of $M$ generated by pairs $\kappa'\kappa''$ with $\kappa$, $\kappa'$ and $\kappa''$ all different. Let $V$ be the set of vertices of the extended dual graph $\hat\Gamma$. 
\begin{warning}
An {\bf extended monopole configuration} is a map
\begin{equation}
m:V\longrightarrow M,
\end{equation}
such that for any color $\kappa$ and any $\kappa$-vertex $v$
\begin{equation}
m(v)\in M_\kappa
\end{equation}
and
\begin{equation}
\sum_{v\in V} m(v) = 0.
\end{equation}
\end{warning}
\noindent
An \emph{extended flux configuration} $\phi$ is any subset of the edges of $\hat\Gamma$. The extended monopole configuration 
\begin{equation}
\hat\partial\phi
\end{equation}
maps a vertex $v$ to the sum of the flux carried by the edges of $\phi$ incident in $v$. The syndrome $\partial\phi$ of a (conventional) flux configuration $\phi$ is the restriction of $\hat\partial\phi$ to the set of inner vertices.

\subsection{Triangle strips}

\begin{figure}
\centering
\includegraphics[width=.6\columnwidth]{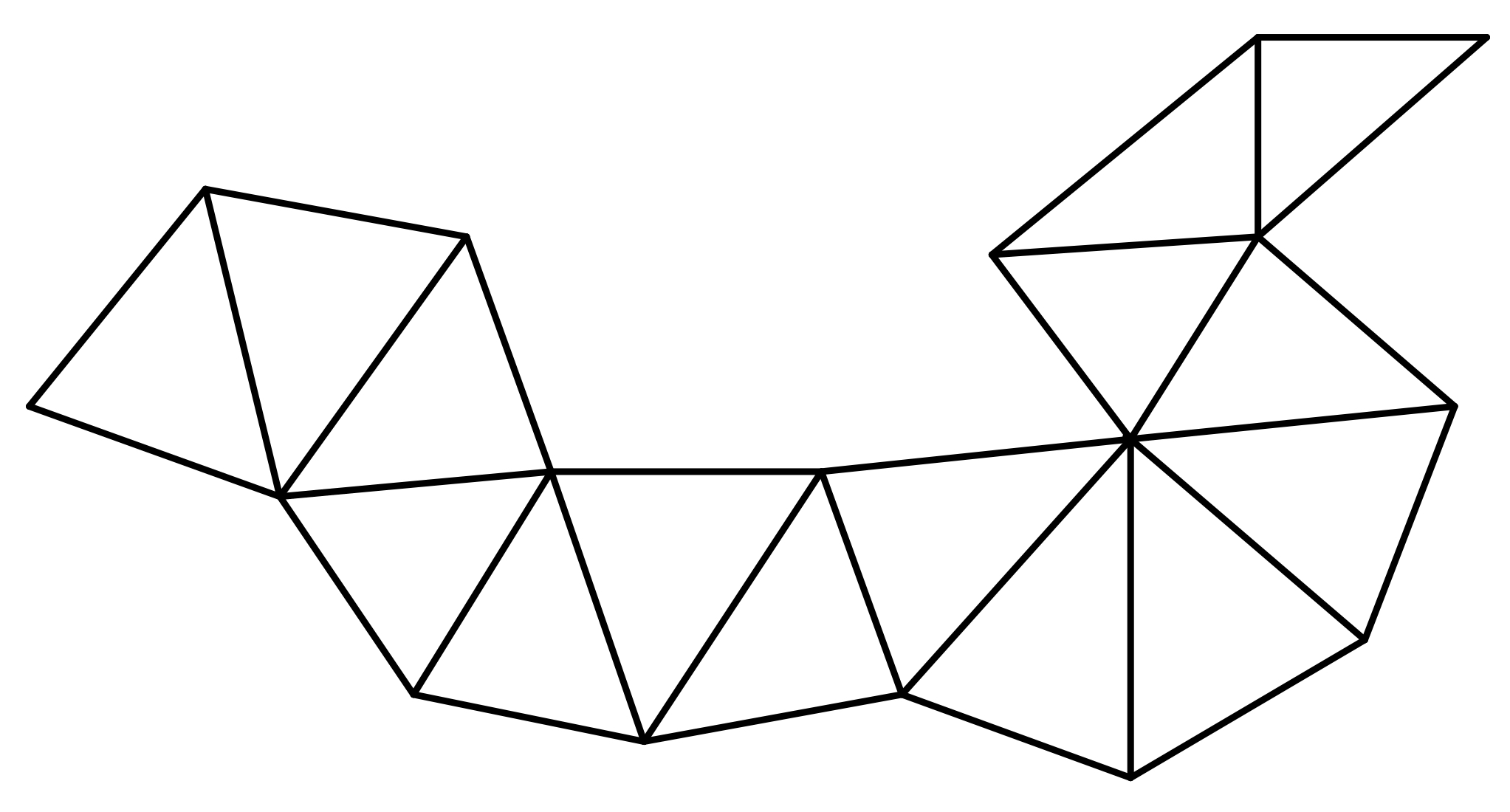}
\caption{A triangle strip.}
\label{fig:strip}
\end{figure}

We want to make use of strips of triangles such as the one in figure~\ref{fig:strip}.
\begin{warning}
A {\bf triangle strip} is a path in the graph that has
\begin{itemize}
\item
the edges of $\hat\Gamma$ as vertices, and \
\item
an edge for each pair of edges of $\hat\Gamma$ that are part of a common triangle.  
\end{itemize}
\end{warning}
We do not directly refer to the (abstract) edges and vertices of such a path, but talk instead only of (the original) triangles and edges. We assume in particular that:
\begin{itemize}
\item
a triangle strip contains at least one edge, and
\item
repeated triangles or edges are allowed.
\end{itemize}

\begin{lem}\label{lem:strip1}
Given 
\begin{itemize}
\item
a triangle strip with vertex set $V$ and edge set $E$, and 
\item
a extended monopole configuration $m$ with support contained in $V$,
\end{itemize}
there exists a extended flux configuration $\phi$ such that
\begin{equation}
\hat\partial \phi = m,\qquad \phi\subseteq E.
\end{equation}
\end{lem}

\begin{proof}[Sketch of proof] 
Lets proceed by induction of the number of triangles. The base case is a single edge, where the result can be checked directly. For the inductive step, consider a strip with $n$ triangles that is composed of (i) a strip $s$ with $n-1$ triangles and (ii) an additional triangle $t$. Consider the vertex $v$ of $t$ that is not an endpoint of the last edge of $s$, and the edges $e_1, e_2$ of $t$ meeting at $v$. The monopole $m(v)$ can take one of four values, each corresponding to a subset $\phi_0\subseteq \{e_1,e_2\}$ via
\begin{equation}
(\partial \phi_0) (v) = m(v).
\end{equation}
Trivially
\begin{equation}
m':=m+\partial \phi_0
\end{equation}
is a extended monopole configuration with support in $s$. By induction there exists $\phi_1$ with its edges in $s$ and such that $\partial \phi_1 = m'$, and thus it suffices to take $\phi = \phi_0+\phi_1$.
\end{proof}

\begin{warning}
A subset of colex faces $F$ is {\bf simple} respect to a set of cells $C$ if, for every cell $c\in C$, the set of points of $c$ that also belong to some element of $F$ is simply connected.
\end{warning}
\noindent

\begin{figure}
\centering
\includegraphics[width=.5\columnwidth]{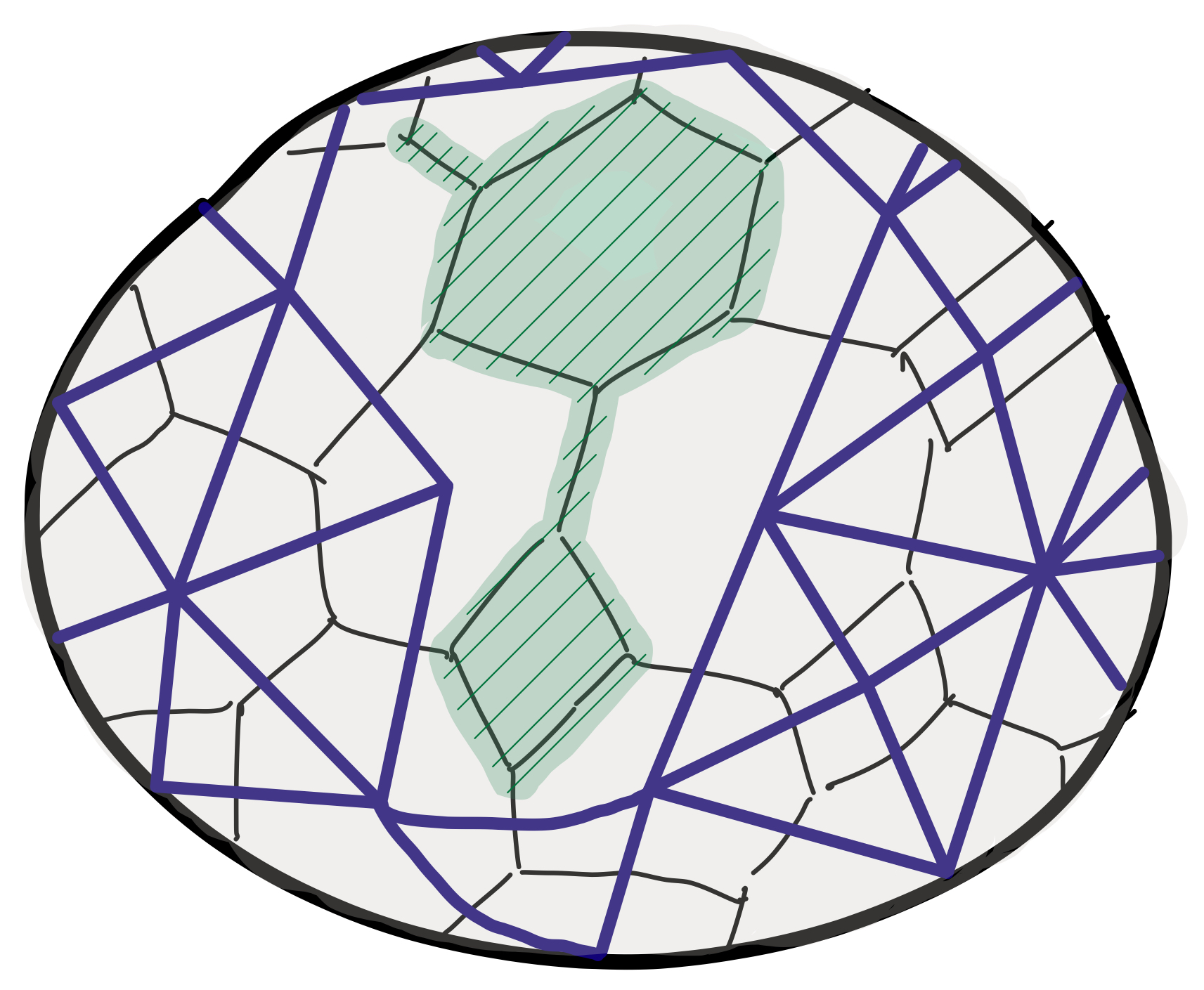}
\caption{
A cell $c$ of a 3-colex. The intersection of $c$ with some set of faces $F$ is indicated in green. The intersection is simply connected and thus compatible with $F$ being simple respect to a set of cells including $c$. The purple edges compose the graph $\gamma$ used in the proof of lemma~\ref{lem:strip2}.
}
\label{fig:simple}
\end{figure}

\begin{lem}\label{lem:strip2}
Let $C$ be a set of cells and $F$ a set of colex faces that is simple respect to $C$. Let $E_0$ be the set of edges of $\hat\Gamma$ that are not dual to faces of $F$. For any path $p$ in $\hat\Gamma$ with
\begin{itemize}
\item
all vertices dual to cells in $C$ except possibly the endpoints, and
\item
edge set $E_p$ such that 
\end{itemize}
\begin{equation}
\emptyset\neq E_p\subseteq E_0,
\end{equation}
there exists a triangle strip with edge set $E$ such that
\begin{equation}
E_p\subseteq E\subseteq E_0,\qquad |E|\leq 1+2(k_\text{face}-1)(|E_p|-1),
\end{equation}
\end{lem}

\begin{proof}[Sketch of proof] It suffices to consider the case with $|E_p|=2$, since the shorter case is trivial and for longer paths we can concatenate the strips obtained for each pair of contiguous edges. Then $E_p=\{e,e'\}$ with $e$ and $e'$ meeting at a vertex $v$ that is dual to a cell $c\in C$. The surface of $c$ is a 2-colex and we can consider its dual graph, which has a vertex per face of $c$, and an edge per edge of $c$. Let $\gamma$ be the subgraph formed by 
\begin{itemize}
\item
the vertices dual to faces not in $F$, and
\item
the edges dual to edges not part of faces in $F$. 
\end{itemize}
This construction is illustrated in figure~\ref{fig:simple}.
By assumption, $\gamma$ is connected. Moreover, 
\begin{itemize}
\item
the vertices in $\gamma$ are in one-to-one correspondence to the edges of $E_0$ with an endpoint at $v$, and 
\item
the edges of $\gamma$ are in one-to-one correspondence to the triangles of $\Gamma$ containing $v$ and with all their edges in $E_0$.
\end{itemize}
Thus there is a path in $\gamma$ of length at most $k_\text{face}-1$ (since it visits each face of $c$ at most once) that provides the desired triangle strip from $e$ to $e'$. The strip has at most $2k_\text{face}-3$ edges that are not in the path $p$.
\end{proof}

\subsection{Closure}

Here we make use of the various definitions of section~\ref{sec:spacetime}. Below we identify the graphs $\Gamma_i$ and $\overline{\Gamma_i}$ with their sets of edges for convenience.

\begin{warning}
$\mathcal O_i$ is the set of flux configurations $\phi\subseteq \Phi_i$ such that the support of $\partial\phi$ is contained in the vertex set of $\overline{\Gamma_i}$. 
\end{warning}
\begin{lem}\label{lem:closure2}
Given some $i$, if 
\begin{itemize}
\item
$\Phi_i$ is simple respect to the set of cells complementary to $C_i$, and
\item
for any vertices $v,v'$ of $\Gamma_i$ and any color $\kappa$
\begin{equation}
d(v,v')\leq k\,\overline d(v,v'), 
\end{equation}
\begin{equation}
\min\left(d_\kappa(v), d(v)\right)\leq k\, \overline {d_\kappa}(v),
\end{equation}
\end{itemize}
then for any $\phi\in\mathcal O_i$ there exists $\phi'\subseteq\overline{\Gamma_i}$ such that
\begin{equation}
\partial \phi = \partial \phi',\qquad |\phi'|\leq 4k(k_\text{face}-1)|\phi|.
\end{equation}
\end{lem}

\begin{proof}[Sketch of proof] Regard such $\phi$ as a subgraph of the extended dual graph $\hat\Gamma$. We assume that $\phi$ is connected (because if it has several connected components $\phi_j$, it suffices to add the corresponding $\phi_j'$) and that $\partial \phi\neq 0$ (because the case $\partial \phi=0$ is trivial, $\phi'=\emptyset$). 

Let $V_\phi$ be the support of $\hat\partial\phi$. Choose a set $V$ containing the inner vertices of $V_\phi$ together with:
\begin{itemize}
\item
two outer vertices of $V_\phi$ of different colors, if $V_\phi$ contains outer vertices of at least two colors,
\item
an outer vertex of $V_\phi$, if $V_\phi$ contains outer vertices of a single color,
\item
no other vertices, otherwise.
\end{itemize}
Choose some tree $t$ that is a subgraph of $\phi$ and has $V$ as its set of leafs (it exists: take any maximal tree of $\phi$ and remove any unwanted leafs repeatedly). By the construction suggested in the figure below (the red dots mark the outer vertices), there exists a collection of paths $p_j$, $j=1,\dots, n$ with edge sets $E_j$ such that 
\begin{itemize}
\item
every $p_j$ has vertices of $V$ as endpoints, 
\item
the last vertex of $p_j$ is the first vertex of $p_{j+1}$,
\item
only the first vertex of $p_1$ and the last of $p_n$ can be outer vertices, and
\item
identifying the tree $t$ with its set of edges:
\end{itemize}
\begin{equation}
t = \bigcup_j E_j,\qquad \sum_j|E_j|\leq 2|t|
\end{equation}
For each $p_j$ we choose some path $p_j'$ in $\hat \Gamma$ with edge set $E_j'\subseteq \overline{\Gamma_i}$ such that
\begin{equation}
|E_j'|\leq k|E_j|
\end{equation}
and, on a case by case basis:
\begin{itemize}
\item
If the endpoints of $p_j$ are both inner vertices, then $p_j'$ has the same endpoints as $p_j$.
\item
If the endpoints of $p_j$ are an inner vertex $v$ and a $\kappa$-colored outer vertex, $p_j'$ is either
\begin{itemize}
\item
a path connecting $v$ to a $\kappa$-colored outer vertex, or
\item
the composition of two paths, each connecting $v$ to an outer vertex, with the colors of these two outer vertices different.
\end{itemize}
\end{itemize}
Since $E_j\subseteq \Gamma_i$ and the inner vertices of $V$ are vertices of $\overline{\Gamma_i}$, such paths exist by assumption. Choose a triangle strip $s_j$ for each $p_j'$ according to lemma~\ref{lem:strip2}. The set of edges
\begin{equation}
E := \bigcup_j E_j'
\end{equation}
is clearly connected, and therefore there exists a triangle strip $s$ such that its triangle set is the union of the triangle sets of the strips $s_j$. Its set of edges $S$ satisfies
\begin{align}
|S|
&\leq 2 (k_\text{face}-1) |E|\leq 2 k (k_\text{face}-1)\sum_j |p_j|
\nonumber\\
&\leq 4k(k_\text{face}-1)|\phi|.
\end{align}
Noting that for any two different colors $\kappa,\kappa'$
\begin{equation}
M = M_\kappa M_{\kappa'},
\end{equation}
it is easy to check (case by case, according to the different possibilities considered above), that there exists a extended monopole configuration $m$ with no support outside the endpoints of the paths $p_j'$ and such that its restriction to the inner vertices is $\partial \phi$.
The result follows applying lemma~\ref{lem:strip1}.
\end{proof}

\subsection{Boundary conditions}\label{sec:open}

Clearly the set $\mathcal C_i$ defined in~\ref{eq:Ci} always satisfies 
\begin{equation}
\mathcal C_i\subseteq \mathcal O_i.
\end{equation}
Therefore, lemma~\ref{lem:closure} follows from lemma~\ref{lem:closure2}.
In fact, when lemma~\ref{lem:closure2} holds, $\mathcal C_i$ takes the form
\begin{equation}
\mathcal C_i = \mathcal O_i,
\end{equation}
which corresponds to truly `open' boundary conditions.

\section{Ball-local noise}\label{sec:ball-local}

This appendix discusses ball-local error distributions, introducing some results that are necessary for the analysis of JIT decoding in appendix~\ref{sec:naive_proof}.

\subsection{Locality and topological codes}\label{sec:locality}

The error correction threshold for a quantum error correcting code is often formulated in terms of local errors, see section~\ref{sec:ft}. Codes are often designed with the expectation that noise is indeed local or approximately so. The aim is for codes to have a high distance, defined as the smallest number of qubits supporting a non-trivial logical operator.

In the case of topological codes another kind of distance enters the picture: \emph{a distance defined by the geometry of the code}, which can be typically codified in a hypergraph: its edges represent qubits, and the support of any non-trivial logical operator has to connect some pair of vertices that are at least separated by some given distance $d$. That is, a non-trivial logical operator can never have support within a ball of radius strictly smaller than $d/2$. This suggests introducing ball-local distributions of errors, as defined in section~\ref{sec:ft_2D}. 

\subsection{From local to ball-local}\label{sec:local_to}

The purpose of this section is to establish a result that is used in section~\ref{sec:naive_result} to show that JIT error correction gives rise to ball-local noise.

Throughout this section it is assumed that some graph is given%
\footnote{
All results apply indistinctly to hypergraphs.
}. 
We say that $v$ is a vertex of an edge set $E$ if $v$ is the endpoint of any edge in $E$.

\begin{warning}
A {\bf ball} is a pair $(v,r)$ with $v$ a vertex, the center of the ball, and $r>0$ an integer, the radius of the ball. Such a ball is identified with the set of edges of the subgraph induced by the set of vertices that are at a distance at most $r$ from $v$.
The radius of a ball $b$ is $\text{r}(b)$ and the sum of the radii of a set of balls $B$ is $\text r(B)$.
\end{warning}
\noindent

\begin{lem}\label{lem:spherification}
Let $\alpha>0$, $c\geq 1$ and $0\leq p\leq 1$. If 
\begin{itemize}
\item
the number of connected subsets of $n$ edges and with a given vertex, is bounded by $\alpha^n$,
\item
the edge set $\omega$ is a random variable satisfying, for any edge set $E$
\begin{equation}
\text{prob}(E\subseteq \omega)\leq p^{|E|},
\end{equation}
\item
the finite set $K$ of connected sets of edges is a function of $\omega$, and for any $\kappa\in K$
\begin{equation}
|\kappa|\leq c |\kappa\cap\omega|,
\end{equation}\end{itemize}
then there exist, for each $\omega$, a ball set $W$, with
\begin{equation}
\bigcup K \subseteq \bigcup W,
\end{equation}
and such that given any ball set $B$
\begin{equation}
\text{prob}\left(B\subseteq W\right)\leq 
\left(\frac p{p_0}\right)^{\text{r}(B)/2c},
\qquad p_0:=(2\alpha)^{-c}.
\end{equation}
\end{lem}
\noindent
We need an auxiliary result.

\begin{lem}\label{lem:S1}
Given a finite set $K$ of connected sets of edges there exist 
\begin{itemize}
\item
disjoint sets $\kappa_i\in K$, and
\item
for each $i$, a ball $b_i$ with center a vertex of $\kappa_i$,
\end{itemize}
such that
\begin{equation}
\text{r}(b_i) = 2|\kappa_i|,
\qquad
\bigcup K \subseteq \bigcup_i b_i.
\end{equation}
\end{lem}

\begin{proof}[Sketch of proof] We proceed by induction on $|K|$. The base case $K=\emptyset$ is trivial. For the inductive step $|K|=n$, assume that the statement holds for any lower cardinality. Choose $\kappa\in K$ with maximal cardinality. The case $\kappa=\emptyset$ is trivial, so we assume that $\kappa$ is not empty. Let 
\begin{equation}
K' := \{\kappa'\in K \,|\, \kappa\cap\kappa' = \emptyset\}.
\end{equation}
Apply the inductive assumption to $K'$ to obtain  $\kappa_i\in K'$ and balls $b_i$, $i=1,\dots, m-1$, with the said properties. Set $\kappa_m = \kappa$, and let $b_m$ be any ball with center a vertex of $\kappa_m$ and radius $2|\kappa_m|$. Clearly all the elements of $K-K'$ are subsets of $b_m$, and $\kappa_m$ overlaps with no $\kappa_i$, $i<m$.
\end{proof}

\begin{proof}[Sketch of proof of lemma~\ref{lem:spherification}] Given $\omega$, choose $W$ as per the prescription of lemma~\ref{lem:S1} applied to the corresponding set $K$. It suffices to show that, given some ball set $B$, the probability that $W$ is a superset of $B$ is bounded as indicated. Let $B= \{b_i\}$. By construction, such an event requires that there exists 

\begin{itemize}
\item
a disjoint collection of sets $\omega_i\subseteq \omega$,
\item
connected sets of edges $\kappa_i$, each with the center of $b_i$ a vertex and 
\end{itemize}
\begin{equation}
\omega_i = \kappa_i\cap\omega,
\qquad 
\text{r}(b_i)=2|\kappa_i|,
\qquad
c|\omega_i|\geq |\kappa_i|.
\end{equation}
Consider the subset of $\omega$
\begin{equation}
\omega' = \bigsqcup_i\omega_i,
\end{equation}
and notice that
\begin{equation}
|\omega'|=\sum_i |\omega_i|\geq \sum_i \frac{|\kappa_i|}{c} = \frac{\text{r}(B)}{2c}.
\end{equation}
For each ball $b_i$ there are at most $\alpha^{\text{r}(b_i)/2}$ different $\kappa_i$ (compatible with $b_i$), and at most $2^{|\kappa_i|}=2^{\text{r}(b_i)/2}$ possible subsets $\omega_i\subseteq\kappa_i$ for each such $\kappa_i$. This gives in total at most $(2\alpha)^{\text{r}(B)/2}$ possible $\omega'$ that can contribute to the event $B\subseteq W$. Each contributes a probability
\begin{equation}
\text{prob}(\omega'\subseteq \omega)\leq p^{|\omega'|}\leq p^{\text{r}(B)/2c}.
\qedhere
\end{equation}
\end{proof}

\subsection{Ball aggregation}

The following results quantifies how unlikely large connected clusters of balls are. The settings are as in section~\ref{sec:local_to}.

\begin{lem}\label{lem:aggregation}
Let $\alpha>0$, $c\geq 1$ and $0\leq p\leq 1$. If 
\begin{itemize}
\item
the number of self-avoiding walks (SAW) of length at most $n$ and starting point on any given vertex is bounded by $\alpha^n$,
\item
the ball set $W$ is a random variable satisfying, for any ball set $B$
\begin{equation}
\text{prob}(B\subseteq W)\leq p^{\text{r}(B)},
\end{equation}
\item
the edge sets $E_i$ are the connected components of $\bigcup W$.
\end{itemize}
Then there exist, for each $W$, balls $b_i$ such that
\begin{equation}\label{eq:Wibi}
E_i \subseteq b_i,
\end{equation}
and the set $A:=\{b_i\}$ is a random variable satisfying for any ball set $B$
\begin{equation}\label{eq:prob_A}
\text{prob}\left(B\subseteq A\right)\leq 
\left(\frac p{p_0}\right)^{\text{r}(B)/2},
\qquad p_0:=\left(\frac 2{3e\alpha}\right)^{2}.
\end{equation}
\end{lem}

\begin{proof}[\skproof]
For each $W$ and each $i$, let $W_i\subseteq W$ be the subsets forming the unique partition of $W$ such that
\begin{equation}
\bigcup W_i = E_i.
\end{equation}
We choose $b_i$ as follows. Consider any two vertices $v, v'$ of $E_i$ with distance equal to the diameter of $E_i$. It is not difficult to check that there exists 
\begin{itemize}
\item
a SAW $q$ from $v$ to $v'$, and 
\item
a ball subset $V\subseteq W_i$ such that $q$ visits all the the centers of the balls in $V$ and has length
\begin{equation}\label{eq:length_q}
|q|\leq 2\text{r}(V).
\end{equation}
\end{itemize}
We set $b_i=(v, 2\text{r}(V))$. 

The relation~\eqref{eq:Wibi} is satisfied because by~\eqref{eq:length_q} the diameter of $E_i$ is bounded by $2\text{r}(V)$. As for~\eqref{eq:prob_A}, given any ball set $B$, the condition $B\subseteq A$ implies that for each ball $b\in B$ there exists such a SAW $q$ and ball set $V\subseteq W$, and these ball sets are mutually disjoint for different elements of $B$. By construction
\begin{equation}
\text{prob}\left(V\subseteq W\right)\leq p^{\text{r}(V)} = p^{\text{r}(b)/2}
\end{equation}
and thus
\begin{equation}
\text{prob}\left(B\subseteq A\right)\leq 
p^{\text{r}(B)/2}\prod_{b\in B} f\left(\frac{\text{r}(b)}2,\text{r}(b)\right)
\end{equation}
where $f(r,l)$ is the maximal number of pairs $(q, V)$ with $V$ a ball set with $\text{r}(V) = r$ and $q$ a SAW that visits the center of each ball in $V$, has length $|q|\leq l$ and has a fixed starting point (the argument of the maximization). For a fixed $q$, the number of possible sets $V$ is bounded by the number of configurations of $r$ identical particles on $l+1$ `states', which equals the number of binary strings with length $r+l$ and weight $l$. That is, we have
\begin{equation}\label{eq:f}
f(r, l)\leq 
\alpha^{l}
\binom {r+l}l
\end{equation}
and the result follows using the bound%
\footnote{
\url{https://en.wikipedia.org/wiki/Binomial_coefficient}
}
\begin{equation}
\binom a b < \left(\frac {ae} b\right)^b.
\qedhere
\end{equation}
\end{proof}

\subsection{Threshold for $\mathbf Z_2$ charge error correction}\label{sec:z2}

The purpose of this section is to show that a well-known class of efficient decoders exhibits an error correction threshold for ball-local noise.

\subsubsection{Decoding a $\mathbf Z_2$ charge}\label{sec:z2_decoding}

The decoding problem of a single $\mathbf Z_2$ charge~\cite{dennis:2002:tqm} is as follows:
\begin{itemize}
\item
a (syndrome) graph represents the code structure,
\item
an error is represented by a subset of edges $E$,
\item
there are \emph{inner} vertices that represent a check operator, and \emph{outer} vertices that do not,
\item
the syndrome $\partial E$ of an error is represented by a set of inner vertices, those that are the endpoint of an odd number of edges in the error set.
\end{itemize}
A decoder outputs a set of edges $E'$ compatible with the error syndrome, \emph{i.e.} such that $\partial E'= \partial E$. The original error $E$ and the decoded error $E'$ can be combined into a logical error $E+E'$ (an error with trivial syndrome). Decoding fails when $E+E'$ is non-trivial, \emph{i.e.} affects the encoded information. 
 
Minimum-weight decoders output a set $E'$ with minimum cardinality among those compatible with the syndrome. There exist efficient implementations~\cite{dennis:2002:tqm}.

\subsubsection{Threshold for ball-local noise}\label{sec:z2_threshold}

The next argument shows that minimum-weight decoding of a $\mathbf Z_2$ charge exhibits an error threshold for ball-local noise. Assume that a family of codes is labeled with an integer $d$ that can be arbitrarily large, so that:
\begin{itemize}
\item
the number of self-avoiding walks (SAW) of length at most $n$ and starting point on any given vertex is bounded by $\alpha^n$,
\item
the number of vertices is polynomial in $d$,
\item
every non-trivial logical error connects two \emph{outer} vertices%
\footnote{Restricting to outer vertices is enough for tetrahedral codes and simplifies the argument.}
with distance at least $d$.
\end{itemize}
We denote by $|p|$ the length of a path $p$. The argument is divided in steps.

\begin{figure}
\centering
\includegraphics[width=.7\columnwidth]{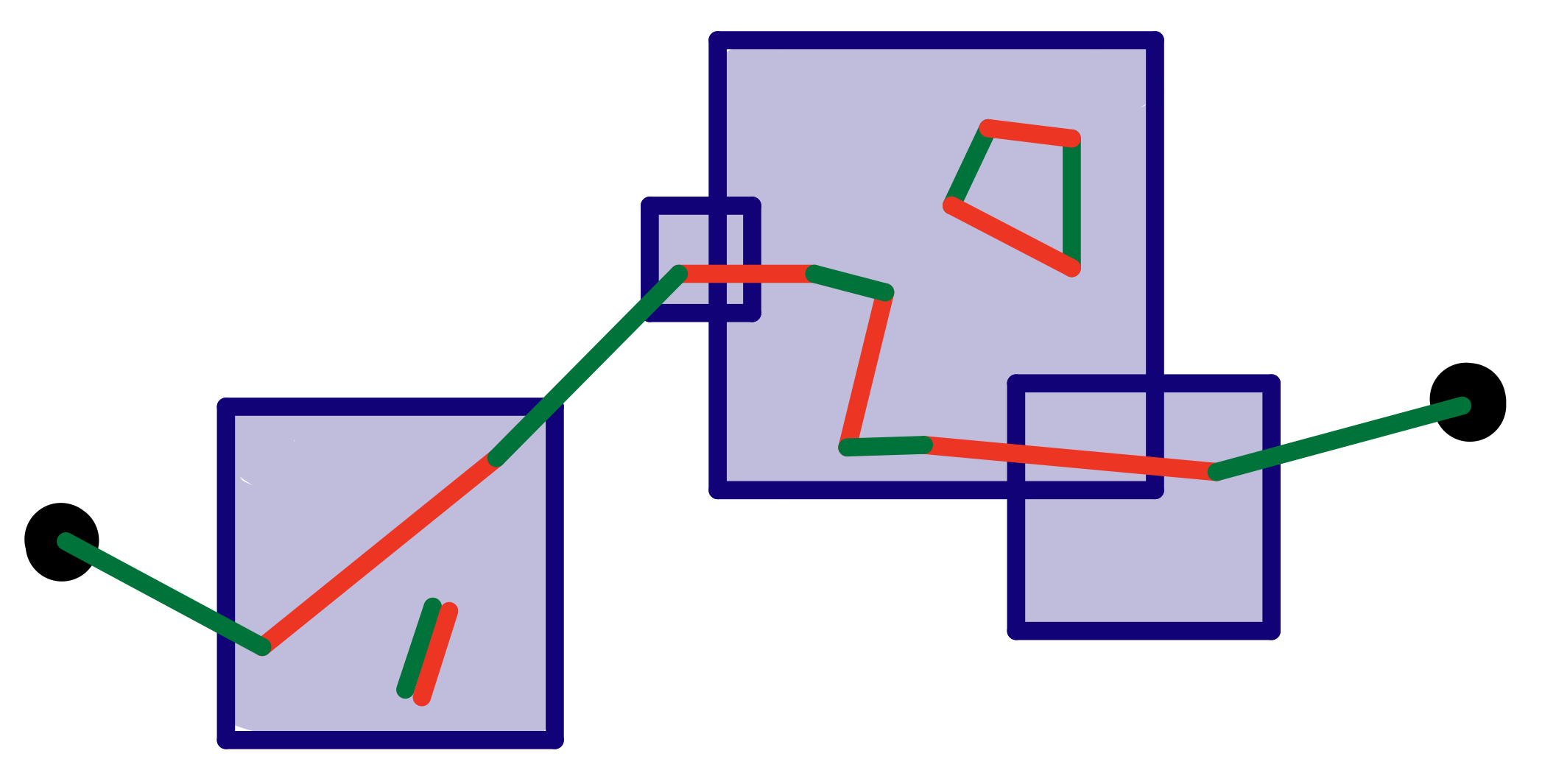}
\caption{
The original error $E$ (red) and the estimated error $E'$ (green) connect two outer vertices (black dots) in the event of a failure. The purple regions represent the balls in the set $W$.
}
\label{fig:bad_error}
\end{figure}

\noindent{\bf 0.} In the event of failure $E+E'$ connects two outer vertices $v$ and $v'$ with distance at least $d$. The error $E$ is a subset of $\bigcup W$, with $W$ as in $\eqref{eq:ball-local}$. This is illustrated in figure~\ref{fig:bad_error}.

\noindent{\bf 1.} It is not difficult to show that there exist ball sets $B_i$, $i=1,\dots, n$, and a SAW $w$ from $v$ to $v'$ obtained by concatenating together a sequence of paths
\begin{equation}\label{eq:walk}
(w_0, w_1', w_1, \cdots, w_n',w_n),
\end{equation}
such that
\begin{itemize}
\item
$B_i\cap B_j = \emptyset$ for $i\neq j$,
\item
$w_i$ has its edges in $E'$,
\item
$w_j'$ visits the center of each ball of $B_j$, and 
\begin{equation}
|w'_j|\leq 2r(B_j).
\end{equation}\end{itemize} 

\begin{figure}
\centering
\includegraphics[width=.7\columnwidth]{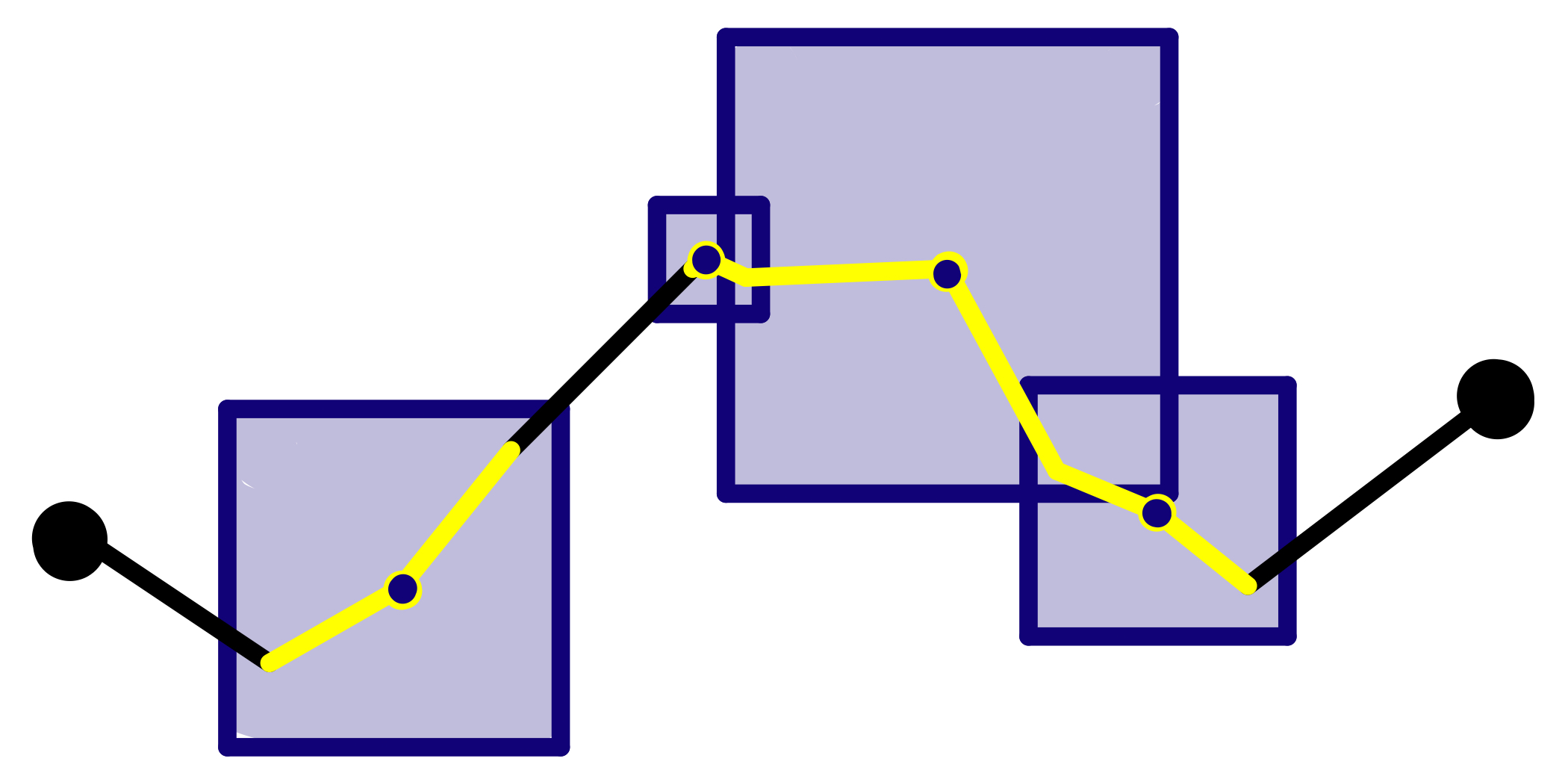}
\caption{
A path as in equation~\eqref{eq:walk}, with the paths $w_i$ in black and the paths $w_i'$ in yellow.
}
\label{fig:path}
\end{figure}

\noindent{\bf 2.} Let $L$ be the set of edges of $w$. Since the endpoints of $w$ are outer vertices
\begin{equation}
\partial (E'+L) = \partial E',
\end{equation}
so that by the minimality of $|E'|$, and noting that by construction $L-E'$ only contains edges of the paths $w_i'$,
\begin{equation}
\sum_{j=1}^n |w_j'|\geq|L-E'|\geq |L\cap E'| \geq \sum_{i=0}^n |w_i|,
\end{equation}
which yields
\begin{equation}
|w|\leq 4r(B),\qquad B:=\bigsqcup_j B_j.
\end{equation}
\noindent{\bf 3.} We have found that a failure event requires of the existence of a SAW $w$ and a set of balls $B\subseteq W$ such that $w$ visits the centers of all the balls in $B$ and
\begin{equation}
d\leq |w|\leq 4r(B).
\end{equation}
The probability of such an event for a fixed set of balls $B$ is bounded by $p^{r(B)}$, and thus the same is true if we fix both $w$ and $B$. The number of such pairs $w$ and $b$ given a fixed starting point for $w$ is bounded by $f(\text{r}(B), 4\text{r}(B))$ with $f$ as in~\eqref{eq:f}. With this observation, the rest of the argument is standard~\cite{dennis:2002:tqm}.

\section{The cost of naivety}\label{sec:naive_proof}

The aim of this appendix is to show that the residual noise $\hat \omega$ of the naive decoder, as described in equation~\ref{eq:JIT_error}, follows a ball-local distribution under reasonable conditions, given that the error rate of the noise afflicting ancilla qubits is below a threshold.

\subsection{Technical conditions on the decoders}\label{sec:minimization}

It is is unclear if minimum weight errors can be computed efficiently for tetrahedral codes.
However, as discussed in~\cite{bombin:2015:single-shot}, it is possible to relax the minimization condition in such a way that efficient decoders exist and certain proof techniques of the minimum weight case can be extended.

The relaxed conditions, adapted to the present scenario, are as follows. 
There exists $k_\text{min}$ such that for every code in the family of tetrahedral color codes of interest and every step $i$ of JIT error correction the open/closed boundary conditions decoders satisfy:

\begin{warning} 
{\bf Minimization conditions}
\begin{itemize}
\item
For any $\omega\subseteq \Phi_i$ and any connected component $\kappa$ of $D_i(\omega)\cup\omega$
\begin{equation}
|\kappa|\leq k_\text{min}|\kappa\cap \omega|.
\end{equation}
\item
For any $\omega\subseteq\overline {\Phi_i}$ and any connected component $\kappa$ of $D_i'(\omega)\cup\omega$
\begin{equation}
|\kappa|\leq k_\text{min}|\kappa\cap \omega|.
\end{equation}\end{itemize}
\end{warning}
\noindent
The notion of connectedness here is given by the dual graph $\Gamma$ of the 3-colex%
\footnote{
With no `outer' vertices, so that dual edges on the boundary have a single endpoint.
}.
The key relationship between connectedness and syndrome is that if $\gamma$ has trivial syndrome, then every connected component of $\gamma$ has trivial syndrome too.
Notice that for decoders that compute minimum weight errors 
\begin{equation}
k_\text{min} = 2.
\end{equation}

\subsection{Ball-locality of the syndrome}\label{sec:naive_result}

The key link between JIT error correction and ball-local noise is lemma~\ref{lem:naive} below. 
Together with lemma~\ref{lem:spherification}, it shows that indeed the residual noise of JIT decoding follows a ball-local distribution (under the given conditions). 
The notation here is as in section~\ref{sec:errors}.

\begin{lem}\label{lem:naive}
If the minimization and closure conditions hold, there exists connected flux configurations $\kappa_j$ such that
\begin{equation}
\hat\omega \subseteq \bigcup_j \kappa_j,
\qquad
|\kappa_j|\leq c |\kappa_j\cap\omega|,
\end{equation}
where
\begin{equation}\label{eq:c}
c=2k_\text{min}(k_\text{min}k_\text{close}+1)+1.
\end{equation}
\end{lem}

\begin{proof} Consider the flux configurations
\begin{equation}
\delta_i := \omega_{i-1}'+\omega'_i\cap\Phi_{i-1},
\end{equation}\begin{equation}
\Omega_i := \epsilon_i\cup\omega_{i-1}'\cup\omega_i'\cup\omega,
\end{equation}
and the partition into connected components
\begin{equation}
\Omega_i = \bigsqcup_j \kappa_{i}^{(j)}.\label{partition}
\end{equation}
Any connected component of $\delta_i$ is an element of $\mathcal C_{i-1}$ because $\delta_i\in \mathcal C_{i-1}$. Since $\delta_i\subseteq \Omega_i$ this implies
\begin{equation}
\kappa_i^{(j)}\cap \delta_i \in \mathcal C_{i-1}.
\end{equation}
Thus by the closure condition there exist
\begin{equation}
\mu_i^{(j)}\in \overline\Phi_{i-1}
\end{equation}
with
\begin{equation}
\mu_i^{(j)}+(\kappa_i^{(j)}\cap \delta_i) \in \mathcal C,
\qquad
|\mu_i^{(j)}|\leq k_\text{close} |\kappa_i^{(j)}\cap \delta_i|.
\label{bound_closure}
\end{equation}
\emph{We assume that each connected component of $\mu_i^{(j)}$ is connected to $(\kappa_i^{(j)}\cap \delta_i)$}, as such a choice is trivially always possible. We perform another partition in connected components, namely
\begin{equation}
\mu_i\cup\Omega_i = \bigsqcup_J \bar\kappa_{i}^{(J)},
\qquad
\mu_i :=\sum_j \mu_i^{(j)}.
\end{equation}
where the indices $J$ form a partition of the set of indices of the partition $\eqref{partition}$:
\begin{align}
\bar\kappa_{i}^{(J)}\cap \Omega_i &= \bigsqcup_{j\in J}\kappa_i^{(j)},
\\
\bar\kappa_{i}^{(J)}\cap\mu_i &= \bigsqcup_{j\in J} \mu_i^{(j)}.
\end{align}
The bound in $\eqref{bound_closure}$ translates into
\begin{equation}
 |\bar\kappa_{i}^{(J)}\cap \mu_i|\leq 
 k_\text{close} |\bar\kappa_i^{(J)}\cap \delta_i|\leq 
k_\text{close} 
|\bar\kappa_i^{(J)}\cap (\omega_{i-1}'\cup\omega_i')|.
\label{bound_closure_2}
\end{equation}
Since
\begin{equation}
\delta_i + \mu_i\in\mathcal C,
\end{equation}
we have, using~\eqref{eq:JIT_epsilon2},
\begin{equation}
\epsilon_i = \delta_i   + \overline{D_{i-1}}(\delta_i ) = \mu_i + D_{i-1}' (\mu_i),
\end{equation}
and thus the minimization condition on $D_{i-1}'$ yields
\begin{equation}
|\bar\kappa_{i}^{(J)}\cap (\epsilon_i\cup\mu_i)|\leq 
k_\text{min} |\bar\kappa_{i}^{(J)}\cap \mu_i|.
\label{bound_min_1}
\end{equation}
Since, for any $i$,
\begin{equation}
\omega'_i = \omega\cap\Phi_i+D_i(\omega\cap\Phi_i),
\end{equation}
the minimazation condition on $D_i$ and $D_{i-1}$ yields
\begin{equation}
|\bar\kappa_{i}^{(J)}\cap (\omega_{i-1}'\cup\omega_i')|\leq 
2k_\text{min} |\bar\kappa_{i}^{(J)}\cap \omega|.
\label{bound_min_2}
\end{equation}
Putting together $(\ref{bound_closure_2}, \ref{bound_min_1},\ref{bound_min_2})$, 
\begin{align}
|\bar\kappa_{i}^{(J)}|
&\leq 
|\bar\kappa_{i}^{(J)}(\epsilon_i\cup\mu_i)|
+|\bar\kappa_{i}^{(J)}\cap (\omega_i'\cup\omega_i)|
+|\bar\kappa_{i}^{(J)}\cap\omega|
\nonumber\\
&\leq
c|\bar\kappa_{i}^{(J)}\cap\omega|,
\end{align}
which is enough because, by $\eqref{eq:JIT_error}$,
\begin{equation}
\hat \omega \subseteq
\bigcup_i\bigcup_J\bar\kappa_i^{(J)}.
\qedhere
\end{equation}
\end{proof}

\subsection{Technical conditions on the lattice}

Theorem~\ref{thm} below relies on the colexes satisfying a set of conditions that ensure that the local structure of the colexes is sufficiently uniform. When the object of interest is not a single tetrahedral code, but rather a family of such encodings (\emph{e.g.} in the study of error thresholds), this uniformity condition should be satisfied by the whole family of colexes.

Let $\mathcal B$ denote the set of all balls in the dual graph $\Gamma$ of a tetrahedral colex, and $\mathcal B'$ the set of all balls in its $X$-error syndrome hypergraph. The edges of this hypergraph are the physical qubits of the code. Thus, every ball in $\mathcal B'$ can be identified with its set of qubits.

\begin{warning}
{\bf Uniformity conditions}
\begin{itemize}
\item
The number of connected subgraphs of the dual graph $\Gamma$ with $n$ edges and containing any given vertex is bounded by $\alpha^n$.
\item
There are constants $m_0, m_1$ and a map 
\begin{equation}
\mathcal M: \mathcal B\longrightarrow \mathcal B'
\end{equation}
such that 
\begin{itemize}
\item
the preimage of any ball in $\mathcal B'$ under $\mathcal M$ has at most $m_0$ elements,
\item
the radius of $\mathcal M(b)$ is $m_1 \text{r}(b)$, and
\item
for any ball $b\in \mathcal B$ and any syndrome $\phi\subseteq b$ there exists $x\in P_X$ with syndrome $\phi$ and support a subset of $\mathcal M(b)$.
\end{itemize}
\end{itemize}
\end{warning}

\noindent
These conditions are satisfied for families of tetrahedral colexes with a uniform and flat local structure, including at facets and corners, such as the one discussed in~\cite{bombin:2015:gauge}. 
The first condition trivially holds if the vertices of the dual graph $\Gamma$ have bounded valence. 
The validity of the second is rooted in the homological structure of the codes, as can be verified using the mapping in~\cite{kubica:2015:unfolding}.

%

\subsection{Residual noise}

The notation here is as in section~\ref{sec:errors}.

\begin{thm}\label{thm}
If the closure, minimization and uniformity conditions hold, and the flux configuration $\omega$ follows a local distribution with rate $p\leq p_0$, with, for $c$ as in~\eqref{eq:c},
\begin{equation}
p_0 = \left( \frac 8 {(3em_0)^4\alpha^5}\right)^c,
\end{equation}
then there exists for each $\omega$ some $x\in P_X$ with syndrome $\hat \omega$ such that $x$ follows a ball-local distribution with rate
\begin{equation}
\left(\frac p{p_0}\right)^{1/4m_1c}.
\end{equation}

\end{thm}

\begin{proof}
By lemmas~\ref{lem:naive}, \ref{lem:spherification} and \ref{lem:aggregation} there exists for each $\omega$ a set of balls 
\begin{equation}
A=\{b_i\}
\end{equation}
and a partition
\begin{equation}
\hat\omega = \bigsqcup_i \hat\omega_i
\end{equation}
such that 
\begin{equation}
\hat\omega_i\subseteq b_i,
\end{equation}
the sets of edges $\hat\omega_i$ are mutually disconnected, and
\begin{equation}
\text{prob}(B\subseteq A)\leq \left(\frac p {p_0'}\right)^{\text{r}(B)/4c},
\qquad p_0' = \left( \frac 8 {(3e)^4\alpha^5}\right)^c
\end{equation}
Since $\hat \omega$ is a syndrome so is each $\hat\omega_i$, and thus there exists $x\in P_X$ with syndrome $\hat\omega$ and support on a subset of the qubits in the balls of $A'=\mathcal M[A]$, the image of $A$ under $\mathcal M$. Each $B'\subseteq \mathcal B'$ defines  a set $m(B')$ containing those ball sets $B\subseteq \mathcal B'$ such that
\begin{itemize}
\item
$\mathcal M[B]=B'$, and
\item
for any $b, b'\in B$ such that $b\neq b'$, $\mathcal M(b)\neq\mathcal M(b')$.
\end{itemize}
The result follows observing that
\begin{equation}
|m(B')|\leq m_0^{|B'|}\leq m_0^{\text{r}(B')/m_1}
\end{equation}
and
\begin{equation}
\text{prob}(B'\subseteq A')\leq \sum_{B\in m(B')} \text{prob}(B\subseteq A).
\qedhere
\end{equation}
\end{proof}

\bibliography{refs}

\end{document}